\documentclass[pdflatex,sn-nature]{sn-jnl}

\usepackage{graphicx}%
\usepackage{multirow}%
\usepackage{amsmath,amssymb,amsfonts}%
\usepackage{amsthm}%
\usepackage{thmtools}
\usepackage{thm-restate}
\usepackage{mathrsfs}%
\usepackage[title]{appendix}%
\usepackage{xcolor}%
\usepackage{textcomp}%
\usepackage{manyfoot}%
\usepackage{booktabs}%
\usepackage{algorithm}%
\usepackage{algorithmicx}%
\usepackage{algpseudocode}%
\usepackage{listings}%
\usepackage{pgfplots}
\pgfplotsset{compat=1.18}
\usepgfplotslibrary{groupplots}
\usetikzlibrary{calc}

\usepackage{cleveref}

\usepackage{mathtools}
\usepackage{array}
\usepackage{tabularx}
\usepackage{wrapfig}
\usepackage{mathtools}

\newcommand{\mcest}[2]{\mathbb{E}_{#1}\left[#2\right]}
\newcolumntype{C}{>{\centering\arraybackslash}X}
\newcommand{\ReIm}{\operatorname{ReIm}}

\renewcommand{\Re}{\operatorname{Re}}

\DeclarePairedDelimiter\ket{|}{\rangle}
\DeclarePairedDelimiter\norm{\lVert}{\rVert}
\DeclarePairedDelimiter\abs{\lvert}{\rvert}

\makeatletter
\DeclarePairedDelimiterX{\braket@int}[2]{\langle}{\rangle}{#1 \delimsize\vert #2}

\NewDocumentCommand{\braket}{s o m}{%
  \kernel@ifnextchar\bgroup
    {\braket@two{#1}{#2}{#3}} 
    {\braket@one{#1}{#2}{#3}} 
}

\NewDocumentCommand{\braket@one}{m m m}{%
  \IfBooleanTF{#1}
    {\braket@int*{#3}{#3}}
    {\IfValueTF{#2}{\braket@int[#2]{#3}{#3}}{\braket@int{#3}{#3}}}%
}

\NewDocumentCommand{\braket@two}{m m m m}{%
  \IfBooleanTF{#1}
    {\braket@int*{#3}{#4}}
    {\IfValueTF{#2}{\braket@int[#2]{#3}{#4}}{\braket@int{#3}{#4}}}%
}
\makeatother

\DeclarePairedDelimiterX{\mel}[3]{\langle}{\rangle}{#1 \delimsize\vert #2 \delimsize\vert #3}

\theoremstyle{thmstyleone}%
\newtheorem{theorem}{Theorem}
\newtheorem{lemma}[theorem]{Lemma}

\theoremstyle{thmstyletwo}%

\theoremstyle{thmstylethree}%

\begin{document}

\title{Projected Inverse Iteration: An Eigenvalue Approach to Ground-State Computation with Neural Quantum States}

\author[1]{\fnm{Hang} \sur{Zhang}}

\author[1]{\fnm{Victor} \sur{Armegioiu}}

\author[1]{\fnm{Juan} \sur{Carrasquilla}}

\author[1]{\fnm{Siddhartha} \sur{Mishra}}

\author[2]{\fnm{Johannes} \sur{Müller}}

\author*[1]{\fnm{Jannes} \sur{Nys}}\email{jannys@ethz.ch}
\equalcont{These authors contributed equally to this work.}

\author*[1]{\fnm{Marius} \sur{Zeinhofer}}\email{marius.zeinhofer@math.ethz.ch}
\equalcont{These authors contributed equally to this work.}

\affil[1]{\orgname{ETH Zurich}, \country{Switzerland}}

\affil[2]{\orgname{TU Berlin}}

\abstract{Deep learning offers a powerful approach to quantum many-body problems via neural network wavefunctions, but their optimization remains a severe bottleneck. Existing optimization methods, including natural gradient descent and stochastic reconfiguration, suffer from spectral gap-dependent convergence that limits their effectiveness on systems fraught with competing orders and nearly degenerate ground states, such as frustrated magnets and strongly correlated electron materials. Here, we introduce Projected Inverse Iteration (PII) by re-framing the ground-state search as an eigenvalue problem. PII achieves rapid, gap-insensitive convergence while preserving the favorable polynomial computational scaling of stochastic reconfiguration. Demonstrated on challenging two-dimensional spin systems, including the highly frustrated $J_1$-$J_2$ model, PII outperforms standard optimization techniques and presents a promising algorithmic strategy for discovering complex quantum states in the presence of small spectral gaps. More broadly, PII can be interpreted as a novel natural gradient method tailored for eigenvalue problems, opening up its application to related challenges within deep learning.
}

\maketitle

\section{Introduction}\label{sec:introduction}

Optimization algorithms are a foundational facet of recent advances in artificial intelligence (AI), alongside novel neural architectures and the scaling of data and compute. They are essential for making large-scale learning practically achievable. Powerful optimization methods range from widely used first-order methods like momentum~\cite{qian1999momentum} and Adam~\cite{kingma2014adam} to computationally heavier approximate second-order methods like K-FAC~\cite{martens2015optimizing}, Shampoo~\cite{gupta2018shampoo}, and SOAP~\cite{vyas2024soap} that offer improved convergence.

This paradigm naturally extends to AI for physics and chemistry. A prominent example is the use of neural network wavefunctions to solve complex many-body problems in condensed matter physics, nuclear physics, and quantum chemistry~\cite{carleo2017solving, wu2024variational, astrakhantsev2021broken, PhysRevResearch.2.023358,mossLeveragingRecurrenceNeural2025, pescia2024message, viteritti2023transformer, chen2024empowering, denis2025accurate, linteau2025phase, pfau2020ab, hermann2020deep, foster2025ab, von2022self, scherbela2024towards, li2024computational, robledo2022fermionic, kim2024neural, lou2024neural, pfau2024accurate}. 
For most architectures, standard first-order stochastic methods such as Adam are well known to perform poorly, making natural gradient schemes (independently discovered and coined stochastic reconfiguration~\cite{sorella1998green}, or SR) the preferred choice. However,
improving this scheme has become a critical focal point; in the last two years, it has become apparent that to achieve breakthroughs, progress in neural architectures must be matched with parallel advancements in optimization methods~\cite{chen2024empowering, rende2024simple, neklyudov2023wasserstein, webber2022rayleigh, goldshlager2024kaczmarz, drissi2024second, peng2025analysis, jiang2025neural, gu2025solving}. Just as Adam revolutionized the field of deep learning, it stands to reason that the development of novel physics-informed optimization schemes may be the key to solving previously intractable quantum problems.

In this manuscript, we focus on the challenge of making \emph{eigensolvers} tractable for large-scale problems in quantum mechanics far beyond the problems reachable with exact methods. In this regime, even state-of-the-art matrix-free methods exhaust the capacity of the world's largest supercomputers. We focus on the core eigenvalue problem of quantum mechanics, the stationary Schrödinger equation
\begin{align}
    \hat{H} \ket{\psi} = E \ket{\psi}. \label{eq:schrodinger}
\end{align}
for a Hamiltonian operator $\hat{H}$ and quantum state $\ket{\psi}$ with energy $E$. While the ground state approximation of Eq. ~\eqref{eq:schrodinger} is typically formulated as an energy minimization problem, here we show that directly viewing it as an eigenvalue problem offers a powerful new perspective. More precisely, our contributions are the following:
\begin{itemize}
    \item \textbf{We introduce Projected Inverse Iteration (PII)}, adapting classical inverse iteration to large-scale quantum systems via a Galerkin projection onto the tangent space of the variational manifold. We empirically demonstrate PII's efficacy on challenging two-dimensional spin systems, where it converges significantly faster than Stochastic Reconfiguration---the current gold standard---while preserving the same favorable computational scaling.
    \item \textbf{We establish a principled theoretical framework} for analyzing the convergence of algorithms in ground-state computations. Our analysis identifies SR as a Galerkin-projected variant of Riemannian gradient descent, explaining its linear convergence rate and severe degradation for closing spectral gaps. In contrast, PII explicitly leverages the eigenvalue structure to bypass this bottleneck, providing a highly robust and efficient training method for neural network wavefunctions in the presence of closing spectral gaps (as illustrated in Figure \ref{fig:overview}).
\end{itemize}

The development of PII in this work can be seen as a specialization of the \emph{infinite-dimensional} viewpoint on optimization in scientific machine learning proposed in \cite{mullerposition}. In this context, infinite-dimensional refers to the choice of inverse iteration as a well-suited method for solving Eq.~\eqref{eq:schrodinger} over a general function space, independent of a parametrization. When we refer to function space or the non-parametric setting, we always mean the intractably large space, irrespective of it being $\mathbb{C}^{2^N}$ in spin systems, or the subspace of $L^2(\mathbb{R}^{3N})$ consisting of anti-symmetric functions in electronic structure.
Consequently, PII applies to general non-linearly parametrized eigenvalue problems. In fact, in machine learning parlance, PII can be viewed as a novel natural gradient method that is tailored for the application to eigenvalue problems. We thus anticipate applications of PII beyond its demonstrated utility in quantum physics.

\begin{figure}
    \centering
    \includegraphics[width=\textwidth]{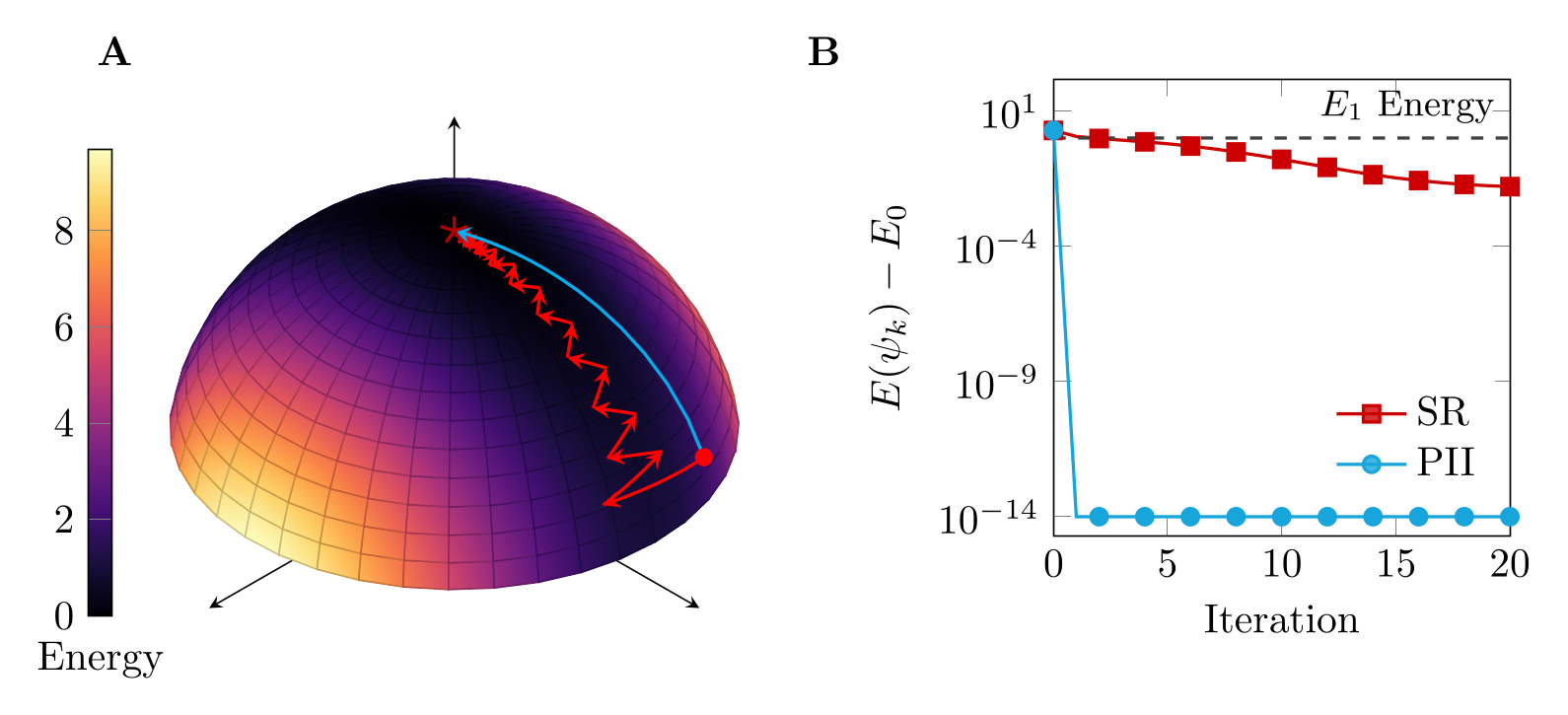}
    \caption{Visualization of SR (red) and PII (cyan) for the Hamiltonian
    $\hat H = \operatorname{diag}(1, 10, 0)$ and a linear normalized state vector on the sphere $\mathbb{S}$. Figure A visualizes the optimization dynamics on the sphere, and Figure B shows the corresponding energy error. PII converges in essentially a single step for the shift value $\tau=10^{-8}$. SR with learning rate $\eta=0.1$ converges slowly with oscillating iterates, due to the large eigenvalue spread ($\Gamma = 10$) and the small spectral gap ($\Delta=1$).}
    \label{fig:overview}
\end{figure}

\section{Results}\label{sec:results}
We introduce the core PII algorithm in the next section, and demonstrate its properties on a benchmark example, as well as on large-scale two-dimensional systems with important Hamiltonians central to condensed matter physics~\cite{wu2024variational}.

\subsection{Projected inverse iteration}
Consider a variational approximation $\psi_\theta$ (typically a neural network~\cite{carleo2017solving}) of a quantum state with parameters $\theta$.
To solve the time-independent Schrödinger equation in Eq.~\eqref{eq:schrodinger}, variational methods cast the problem into an optimization problem of the energy loss function given by the Rayleigh quotient
\begin{equation}\label{eq:energy}
    \min_{\psi_\theta}E(\psi_\theta) 
    =
    \frac{\mel{\psi_\theta}{\hat{H}}{\psi_\theta}}{\braket{\psi_\theta}}.
\end{equation}
The denominator translates the fact that valid quantum states are normalized ($\ket{\smash{\hat{\psi}}}=\ket{\psi}/\norm{\psi}$ in the $L_2$ norm $\norm{\psi} = \sqrt{\braket{\psi}}$), i.e.\ they live on the unit sphere $\mathbb{S} = \{\hat{\psi} \in \mathcal{H} : \norm{\smash{\hat{\psi}}}=1\}$. Hence, the Rayleigh quotient in Eq.~\eqref{eq:energy} is insensitive to the normalization.
The above objective is typically optimized using gradient-based optimization algorithms such as Adam~\cite{kingma2014adam,PhysRevResearch.2.023358}, K-FAC~\cite{martens2015optimizing, pfau2020ab}, or stochastic reconfiguration (SR). While successful in many contexts, these minimization methods do not fundamentally exploit the underlying eigenvalue problem represented by the Schrödinger equation in Eq.~\eqref{eq:schrodinger}.

To respect the eigenvalue structure, we design a projected version of inverse iteration---the canonical algorithm in numerical linear algebra for computing eigenstates given an eigenvalue guess \cite{trefethen2022numerical}. The projection step turns the intractable non-parametric inverse iteration scheme into a scalable iterative algorithm similar to a preconditioned gradient scheme. To translate this approach to the context of ground-state computations in quantum mechanics, we consider a trial $\ket{\psi_k}$ at iteration $k$ and a shift parameter $\tau$ that serves as an approximation of the ground-state energy. In a non-parametric setting, inverse iteration updates the quantum state via
\begin{equation}
    \ket{\hat{\psi}_{k+1}} 
    =
    \frac{(\hat{H} - \tau \hat{I})^{-1}\ket{\hat{\psi}_k}}{\norm{(\hat{H} - \tau \hat{I})^{-1}\ket{\hat{\psi}_k} }}.\label{eq:inverse_iteration_functional_update}
\end{equation}
We rewrite the above into a Riemannian scheme, iteratively updating the current state with tangential directions and retracting back to the sphere via normalization. This yields
\begin{align}
    \ket{\hat{\psi}_{k+1}} 
    =
    \frac{\ket{\hat{\psi}_k} - \ket{d_k}}{\norm{\ket{\hat{\psi}_k} - \ket{d_k}}}, 
    \quad \ket{d_k}
    =
    \ket{\hat{\psi}_k}
    -
    \frac{(\hat H - \tau \hat{I})^{-1}\ket{\hat{\psi}_k}}{\mel{\hat{\psi}_k}{(\hat H - \tau \hat{I})^{-1}}{\hat{\psi}_k}}.\label{eq:tangent_inverse_iter_step_non_parametric}
\end{align}
where $d_k\in T_{\psi_k}\mathbb S$ is tangent to the current iterate $\hat{\psi}_k$, i.e.\ $\braket{d_k}{\hat{\psi}_k} = 0$. 
In the example in Fig.~\ref{fig:overview}, we explicitly depict the update vector $d_0$ on the sphere for the first step ($k=0$). At first glance, this reformulation appears impractical. Direct computation of $d_k$ is infeasible for many-body systems due to the exponential growth of the Hilbert space, and introducing Riemannian geometry may seem like an unnecessary complication. However, this geometric perspective is the crucial stepping stone for deriving a scalable, projected variant. Specifically, our method builds on the realization that a Galerkin projection of this infinite-dimensional Riemannian update onto the tangent space of a variational model leads to a highly tractable algorithm. Moreover, this approach maintains the same polynomial computational scaling as SR per iteration (for details, see Methods).
Galerkin projection turns the inverse iteration update~\eqref{eq:tangent_inverse_iter_step_non_parametric} into a tractable preconditioned gradient scheme in parameter space
\begin{equation}
    \theta_{k+1}
    =
    \theta_k
    -
    Q^{-1}(\theta_k)\left[\tfrac12 \nabla_\theta E(\psi_{\theta_k})\right],\label{eq:preconditioned_gradient_update_PII}
\end{equation}
where $\nabla_\theta E(\psi_{\theta_k})$ is the energy gradient, and $Q$ is a preconditioner that can be estimated efficiently with Monte Carlo integration. While we use normalized states in our notation, our method works equally well for unnormalized states. For details, we refer to the Methods. The matrix $Q$ depends on the energy shift $\tau$, and is constructed from the $\tau$-shifted Hamiltonian operator projected into the tangent space of the current iterate $\hat{\psi}_k$. Moreover, for the choice $\tau=E(\psi_k)$, the iteration \eqref{eq:inverse_iteration_functional_update} is known as \emph{Rayleigh quotient iteration}, and its equivalent Riemannian formulation in Eq.~\eqref{eq:tangent_inverse_iter_step_non_parametric} agrees with the Riemannian Newton method. We expand on this connection in Appendix~\ref{inverse_iter_as_newton}.

\subsection{Convergence guarantees}
To analyze PII in comparison to SR, we consider their non-parametric formulations. This setting isolates the underlying dynamics and has the advantage of being agnostic to the specific parametric ansatz. Extending this analysis to include neural quantum states and the effect of Monte Carlo sampling remains an open question beyond the scope of this work. Strictly speaking, our computationally tractable algorithms in parameter space are Galerkin-projected and Monte Carlo-estimated variants of their functional counterparts. While we refer to them as SR and PII for simplicity, their underlying non-parametric methods are Riemannian gradient descent and inverse iteration, see also Methods.

Let $E_0$ and $E_{\text{max}}$ denote the smallest and largest eigenvalues of $\hat H$, respectively. We define the first excited state energy, $E_1$, as the minimal eigenvalue of $\hat H$ strictly greater than $E_0$. Our analysis explicitly covers degenerate ground states.
For any iterative sequence $\hat{\psi}_k$ converging to a target ground state $\hat{\psi}^*$, we analyze the functional convergence rate $\rho < 1$ for which  $\norm{\hat \psi_{k} - \hat \psi^*}_2 \leq c \rho^k$ for a constant $c$ independent of $k$. 
The convergence of SR as a functional algorithm is governed by the rate
\begin{equation}
    \rho_{\text{SR}} 
    =
    \left| 
        1 - \frac12\frac{E_1 - E_0}{E_{\text{max}} - E_0}
    \right|.\label{eq:rho_SR}
\end{equation}
For concreteness this rate is stated for $\eta = 1/(2(E_{\text{max}} - E_0))$. Crucially, the convergence rate of SR in Eq.~\eqref{eq:rho_SR} deteriorates when the spectral spread $\Gamma = E_{\text{max}} - E_0$ is large or the spectral gap $\Delta = E_1 - E_0$ is small. In these cases, SR converges with a slow linear rate, as illustrated in Figure~\ref{fig:overview} or in more detail in the Appendix~\ref{sec:graphical_visualization}. Moreover, SR possesses a critical step-size threshold of $\eta=1/(E_{\text{max}} - E_0)$ beyond which the algorithm diverges, see Figure [B]. The sensitivity of the SR step-size to the spectrum ($\eta<\Gamma^{-1}$) requires a manual step-size tuning for every quantum system. 

Inverse iteration, on the other hand, always uses a step-size of $\eta=1$ and converges for all shift parameters $\tau < \frac{E_1 + E_0}{2}$ at the rate
\begin{equation}
    \rho_{\text{II}}
    =
    \left|
        1 - \frac{E_1 - E_0}{E_1 - \tau}
    \right| = \left|\frac{E_0 - \tau}{E_1 - \tau}\right|. \label{eq:rho_II}
\end{equation}
For $\tau = E_0+\alpha (E_1-E_0)$ we obtain the convergence rate $\rho_{\textup{II}} = \frac{\alpha}{1-\alpha}$ thereby removing all dependency on the spectrum of $\hat H$. 
Further, the convergence factor $\rho_{\text{II}}$ in Eq.~\eqref{eq:rho_II} can be made arbitrarily fast as $\tau\to E_0$, see Figure~\ref{fig:gap_sensitivity}.  
In the limit $\tau \to E_0$, inverse iteration converges typically in a single step as illustrated in Figure~\ref{fig:overview}. 
Nevertheless, using a value $\tau \neq E_0$ induces a slowdown of convergence, which we demonstrate in Figure~\ref{fig:undershoot}. However, we will demonstrate in the next sections that $\tau < E_0$ can serve as a reliable regularization when used with stochastic Monte Carlo estimation of the quantities involved. When $\tau_k = E(\psi_k)$ is chosen adaptively as the current energy estimate, the method is known as the eigensolver Rayleigh quotient iteration and converges cubically in the vicinity of the ground state~\cite{trefethen2022numerical}. While the convergence analysis in this section is for the non-parametric setting, we experimentally show that these properties to a large extent carry over to expressive variational models. In practice, this means PII is much more robust to small gaps compared to SR. In all our experiments in the next section, including linear models, full Hilbert sum computations, and highly nonlinear deep neural networks on large $10 \times 10$ models with Monte Carlo, we clearly observe significantly less gap sensitivity in agreement with the theory. 

In certain applications, extremely small spectral gaps can occur. For any method to resolve these, one needs access to highly accurate energies, gradients, and preconditioners beyond those achievable with Monte Carlo estimation. Resolving such gaps remains out of reach for both SR and PII. However, PII can resolve moderately small gaps much better than SR.

\begin{figure}[t]
  \centering
  \includegraphics{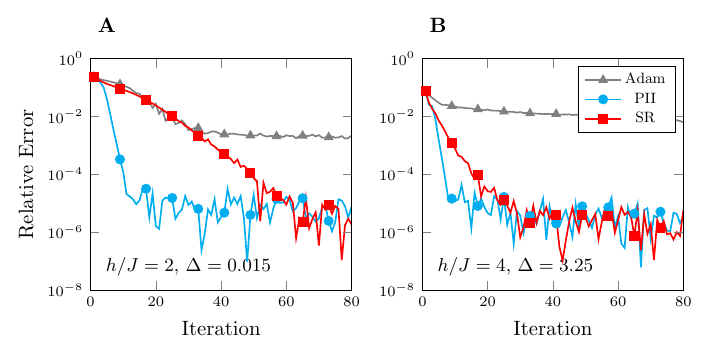}
  \caption{\textbf{Spectral-gap dependence distinguishes PII from SR.}
  \textbf{A,} Optimization trajectories for the $4\times4$ TFIM in a regime with a small spectral gap ($h/J=2$,  $\Delta=0.015$), plotted as relative error $\varepsilon$ versus iteration for PII, SR and Adam, where $\varepsilon = \left|\left(E_{\mathrm{ViT}} - E_{\mathrm{ED}}\right)/E_{\mathrm{ED}}\right|$. All runs use the same number of Monte Carlo samples ($M=32{,}768$) and a positive real-valued ViT ansatz (4 layers, hidden dimension $d=20$, and $h=2$ heads). PII uses a learning rate $\eta=1.0$ and Tikhonov regularization $\lambda=0.01$, whereas SR uses $\eta=0.01$ and $\lambda=10^{-4}$. Adam is included as a standard first-order baseline, with its best learning rate $\eta=0.01$.
  \textbf{B,} The same comparison in a larger-gap regime ($h/J=4$, $\Delta=3.25$). PII shows little gap dependence, whereas SR slows down for the smaller gap and Adam remains at higher relative error over the plotted iterations.}

  \label{fig:spectral_gaps}
\end{figure}

\subsection{Effect of spectral gap on convergence}

In this section, we experimentally confirm the favorable convergence properties of PII on systems with small spectral gaps, both for a linearized state-vector setting and a nonlinearly parametrized neural-network ansatz. For the linear state vector, we visualize the behavior of SR for a moderately small spectral gap in Fig.~\ref{fig:overview}. The effect of further shrinking the gap is visualized in Appendix~\ref{sec:graphical_visualization} in Figures \ref{fig:gap_sensitivity} and \ref{fig:stability}. We observe that closing spectral gaps lead to a slow oscillatory behavior of the iterates in Hilbert space. This is the familiar behavior of a gradient descent method in a narrow loss valley, bouncing rapidly between the valley walls. This interpretation is further strengthened through the analysis of SR as a projected Riemannian gradient descent method in section \ref{sec:riemannian_interpretation}. PII, on the other hand convergence in a single step for all these examples, independent of the gap sizes considered.

We further benchmark PII against SR and Adam in a parametric setting using a variational ansatz for a small $4 \times 4$ transverse-field Ising model (TFIM), for which the spectral gap can be evaluated by exact diagonalization. The Hamiltonian is
\begin{align}
\hat{H}_{\mathrm{TFIM}}=-J \sum_{\langle i, j\rangle} \hat{\sigma}_i^z \hat{\sigma}_j^z-h \sum_i \hat{\sigma}_i^x, 
\end{align}
where $\langle i, j\rangle$ denotes nearest neighbours, and $\hat{\sigma}^{x, y, z}$ are the Pauli matrices. 

Figure~\ref{fig:spectral_gaps} compares the convergence speed in a coupling regime with different gap sizes $\Delta = E_1 - E_0$. We consider the gaps (small)  $\Delta =0.015$ at $h / J=2$, and (large) $\Delta=3.25$ at $h / J=4$, respectively. The experiments are carried out with a deep neural network ansatz, in particular a non-linear ViT ansatz~\cite{viteritti2023transformer}. We include Adam as a widely used adaptive first-order baseline, which relies only on local gradient information and does not use the preconditioning employed by SR or PII. For SR and PII, the respective preconditioners are regularized with a Tikhonov regularization\footnote{We will set the units of energy $J=1$ everywhere, and drop the scale in SR learning rates and PII regularizations for simplicity of notation.} $Q(\theta) \to Q(\theta) + \lambda \mathbb{I}$ when taking the inverse in Eq.~\eqref{eq:preconditioned_gradient_update_PII} (see Eq.~\eqref{eq:regularization} in Methods). For fair comparison, we run SR with the largest possible learning rate $\eta$ and the smallest possible regularization $\lambda$ that allows for a stable optimization (see Supplemental Material).

At $h / J = 2$, the $4 \times 4$ TFIM corresponds to the ordered ferromagnetic regime of the two-dimensional model. In the thermodynamic limit, this phase spontaneously breaks the global $\mathbb{Z}_2$ spin-flip symmetry, giving two symmetry-related ground states. On a finite periodic lattice, this degeneracy is lifted by a finite-size spectral gap $\Delta$, which decreases exponentially with system size. PII reduces the relative error of the variational energy to $\sim 10^{-4}$ within $10$ iterations in both regimes (Fig.~\ref{fig:spectral_gaps}a,b). By contrast, SR slows markedly as the gap decreases, requiring substantially more iterations to reach the same accuracy. The resulting trajectories distinguish an almost gap-insensitive optimization regime for PII from a gap-limited regime for SR. These results indicate that the qualitative conclusions based on theoretical non-parametric convergence rates in Eqs.~\eqref{eq:rho_SR}-\eqref{eq:rho_II} persist for neural quantum states optimized with Monte Carlo sampling.

\subsection{Effect of the energy shift $\tau$}\label{sec:undershoot}

\begin{figure}[tb]
  \centering
  \includegraphics{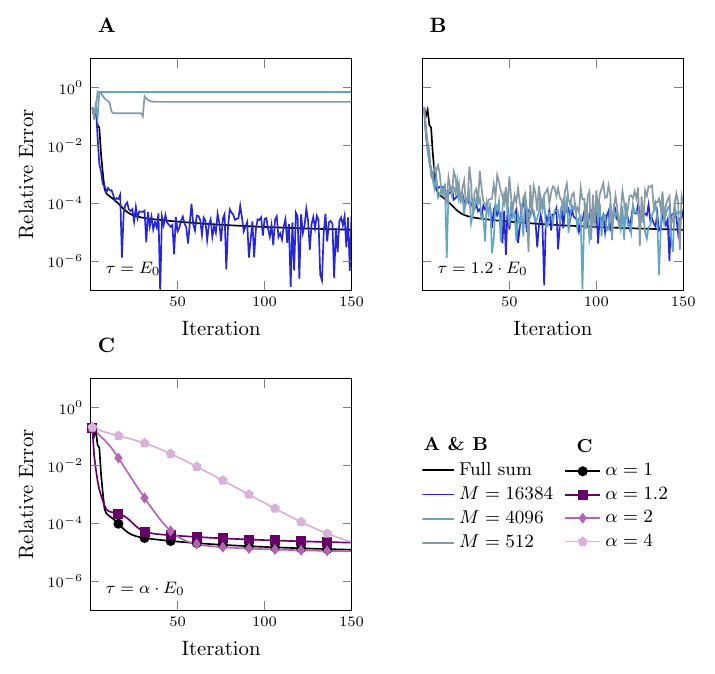}
  \caption{\textbf{Shift undershooting stabilizes PII in a quasi-degenerate TFIM.}
  Relative ground-state energy error versus iteration for the $4\times4$ TFIM at $h/J=2$.
  \textbf{A,} Without undershooting ($\tau=E_0$), Monte Carlo noise prevents PII from converging for small batch sizes ($M=512$ and  $M=4096$).
  \textbf{B,} With undershooting ($\tau=1.2\,E_0$),  trajectories stay stable across batch sizes.
  \textbf{C,} In the noise-free full-sum limit, sweeping $\tau=\alpha E_0$ follows the predicted contraction $r_{\mathrm{PII}}$, showing a direct trade-off between robustness and speed. 
  All trajectories use $\lambda=0.1$ and $\eta=1.0$. All results are produced with a real-valued RBM ansatz (hidden-variable density $\alpha=4$~\cite{carleo2017solving}).}
  \label{fig:undershoot_batch_sizes}
\end{figure}

Theoretically, the optimal convergence in PII is expected when the shift $\tau$ matches the true ground state, i.e., $\tau = E_0$ in  Eq.~\eqref{eq:inverse_iteration_functional_update}. However, taking $\tau < E_0$, which we refer to as \emph{undershooting} the energy, can help stabilize the preconditioning at a fixed Monte Carlo budget per iteration, at the cost of a reduced theoretical convergence rate in Eq.~\eqref{eq:rho_II}. This approach is different from the standard Tikhonov regularization approach ($\lambda$) also used in SR.

In Fig.~\ref{fig:undershoot_batch_sizes}a,b, we consider the 2D TFIM at $h/J=2$ on a $4\times4$ lattice. Without undershooting, i.e.\ $\tau=E_0$ is the exact ground state energy obtained with ED (Fig.~\ref{fig:undershoot_batch_sizes}A), PII shows rapid convergence, but requires a higher number of Monte Carlo samples for convergence. For $M=512$ and $4096$, the optimization fails, likely to an inaccurately empirical preconditioner in Eq.~\eqref{eq:preconditioned_gradient_update_PII}. In contrast, at $M=16384$ the method converges smoothly, indicating effective preconditioning. 

Introducing mild undershooting substantially improves the robustness of PII (Fig.~\ref{fig:undershoot_batch_sizes}B): across the same range of $M$, PII becomes far less sample-size dependent and converges reliably even in the low-sample regime. We attribute this stabilization to improved spectral conditioning of the preconditioner. Increasing $\tau$ shifts the spectrum of $Q$ towards larger values, which heuristically makes it easier for finite-sample estimates to retain (approximate) positive semi-definiteness in Eq.~\eqref{eq:Q_matrix}.

To disentangle this effect from Monte Carlo noise, we further evaluate PII in a noiseless setting by replacing stochastic estimators with a full Hilbert space contraction (full sum for short). In this deterministic limit (see Fig.~\ref{fig:undershoot_batch_sizes}c), undershooting is no longer expected to provide a stability advantage, but still induces the slowdown of convergence observed in Eq.~\eqref{eq:rho_II}. As $\tau$ moves away from $E_0$, the ratio approaches unity, increasing $\rho_{\mathrm{II}}$ and weakening the per-iteration contraction. These results expose a stability-speed trade-off that can be exploited in PII: undershooting mildly reduces the convergence speed, but significantly improves numerical robustness under stochastic estimation of the preconditioner $Q$.

\subsection{Large scale results}

\begin{figure}[tb]
    \centering
    \includegraphics{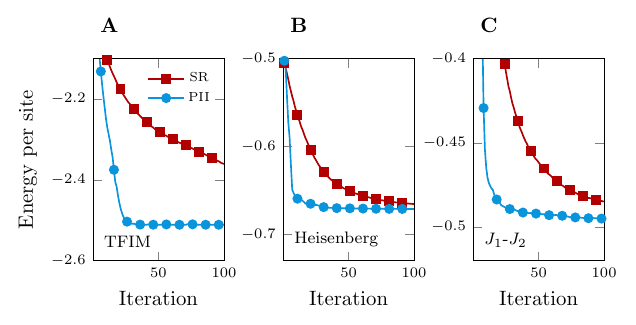}
    \caption{\textbf{PII acceleration at large $10\times10$ lattices.}
  Energy per site during the optimization, comparing SR and PII for three canonical 2D spin-$\tfrac{1}{2}$ models on a torus:
  \textbf{A,} TFIM at $h/J=2$;
  \textbf{B,} unfrustrated Heisenberg antiferromagnet ($J_2=0$);
  \textbf{C,} frustrated $J_1$--$J_2$ model at $J_2/J_1=0.5$. Stoquastic cases \textbf{A,B} are optimized with a positive real ViT ansatz (4 layers; $d=20$, $h=2$), whereas the frustrated point \textbf{C,} uses a complex-valued ViT to capture the emergent sign structure (4 layers; $d=60$, $h=10$; see Methods). PII uses $\eta=1.0$ with diagonal shift $10^{-3}$ (TFIM), $\eta=1.0$  with shift $10^{-2}$ (Heisenberg), and $\eta=1.0$ with shift $10^{-1}$ at $J_2/J_1=0.5$, together with mild overshoot control.}
    \label{fig:10x10_systems}
\end{figure}

We now apply PII to a number of large-scale experiments far beyond the reach of exact methods (due to a Hilbert-space size of $10^{30}$), and of significant importance for condensed matter physics~\cite{wu2024variational}. We consider three canonical two-dimensional spin-$\frac{1}{2}$ models on $10\times10$ square lattices with periodic boundary conditions: the TFIM, the (unfrustrated) Heisenberg antiferromagnet, and the highly frustrated $J_1$--$J_2$ model. The latter reads
\begin{align}    
\hat{H}_{J_1\text{--}J_2}=J_1 \sum_{\langle i, j\rangle} \hat{\boldsymbol{S}}_i \cdot \hat{\boldsymbol{S}}_j+J_2 \sum_{\langle\langle i, j\rangle\rangle} \hat{\boldsymbol{S}}_i \cdot \hat{\boldsymbol{S}}_j \label{eq:j1j2},
\end{align}
where $\langle\langle i, j\rangle\rangle$ denotes next-nearest neighbours, and $\boldsymbol{S}_i = [\hat{\sigma}_i^x, \hat{\sigma}_i^y, \hat{\sigma}_i^z]$. The Heisenberg model corresponds to $J_2/J_1=0$ in Eq.~\eqref{eq:j1j2}. Near the frustrated point $J_2/J_1=0.5$, the phase diagram in the thermodynamic limit remains highly debated~\citep{Jiang2012PRB,Gong2014PRL,Ferrari2020PRB,Richter2014arXiv}.

The TFIM and the Heisenberg Hamiltonian are stoquastic, i.e.\ can be cast into a form with real and negative off-diagonal elements. This implies that their ground states can be chosen real and nonnegative. In contrast, near $J_2/J_1 \simeq 0.5$ the $J_1$--$J_2$ model develops a nontrivial sign structure and is challenging even for state-of-the-art ground-state solvers. Capturing the sign structure requires an expressive and highly non-linear complex-valued ansatz $\psi_\theta$.
For TFIM at $h/J=2$ and the Heisenberg limit $J_2/J_1=0$, we use the same positive ViT as in previous experiments. For the highly frustrated point $J_2/J_1=0.5$, we use a wider and complex-valued neural network.

Figure~\ref{fig:10x10_systems} shows the convergence of PII versus standard SR for the three Hamiltonians. PII reaches the target energy in substantially fewer iterations than SR for the TFIM and Heisenberg benchmarks (Fig.~\ref{fig:10x10_systems}a,b), which is consistent with SR becoming spectrum-limited in systems with small energy gaps. On $10\times10$, SR typically requires $\mathcal{O}(10^2)$ iterations to reach the target energy, whereas PII converges within only a few dozen iterations ($\sim 20$--$30$) across the same settings.

We clarify that, in the ordered TFIM, the finite-size splitting between the two lowest states becomes exponentially small with system size. Because these states form a quasi-degenerate ground state, resolving them with single-state methods is practically impossible. Consequently, convergence to the target energy is dominated by the gap to the next true excited state. When this gap is small, SR enters a gap-limited regime. Similar effects appear in the Heisenberg antiferromagnet, while in the $J_1$-$J_2$ model, frustration lifts the quasi-degeneracy and can remain significant for large system sizes: $\mathcal{O}(1)$.

\section{Methods}\label{sec:methods}
\subsection{Riemannian functional optimization perspective}\label{sec:riemannian_interpretation}
Before describing the projection step that converts II into PII, we first analyze energy optimization from a Riemannian perspective. We hereby focus on infinite-dimensional functional updates, or said differently, updates of non-parametric quantum states in the Hilbert space. The Riemannian gradient of the energy $E(\hat{\psi}) = \mel{\hat{\psi}}{H}{\hat{\psi}}$ on the sphere $\mathbb{S}$ reads
\begin{align}
    \operatorname{grad}_{\mathbb{S}} E (\hat \psi)
    =
    2(\hat{H} - E(\hat{\psi}) \hat{I})\ket{\hat{\psi}}
\end{align}
Hence, Riemannian gradient descent on the sphere with step-size $\eta>0$ and retraction by normalization \cite{boumal2023introduction} is given by the recursion
\begin{align}
    \ket{\hat{\psi}_{k+1}} 
    &=
    \frac{\ket{\hat{\psi}_k} - \eta\operatorname{grad}_\mathbb SE(\hat{\psi}_k)}{\norm{\ket{\hat{\psi}_k} - \eta\operatorname{grad}_\mathbb SE(\hat{\psi}_k)}}. \label{eq:SR_Riemannian}
\end{align}
In other words, the Riemannian gradient descent update in Hilbert space retracts the following update direction
\begin{align}
    \begin{split}
        \ket{d_k} 
        &= \eta\operatorname{grad}_\mathbb{S} E(\hat{\psi}_k) \\
        &= 
        2\eta (\hat H   - E( \hat{\psi}_k) \hat{I}) \ket{\hat{\psi}_k}
    \end{split}
\end{align}
back onto the sphere $\mathbb S$.
We will demonstrate in the next section that SR can be understood as a Galerkin-projected version of this Riemannian gradient descent scheme.

We note in passing that the Riemannian energy gradient is the same as the gradient of the Rayleigh quotient in Eq.~\eqref{eq:energy}: $\operatorname{grad}_{\mathbb{S}} E = \nabla_{\psi} E$. This implies that Riemannian gradient descent agrees with gradient descent on the Rayleigh quotient up to the retraction step.

\subsection{Galerkin projections and inverse iteration}
While the previous section considered functional updates, we now study these updates in a basis determined by a parametrized quantum state.
For a normalized neural network wavefunction $\hat{\psi}_\theta=\psi_\theta/\|\psi_\theta\|$ with parameters $\theta\in\mathbb R^P$, we consider its local linearization given by
\begin{align}
    \mathcal V_\theta 
    =
    \operatorname{span} \{ \partial_{\theta_i} \hat{\psi}_\theta \}_{i=1}^P 
    \subset 
    T_{\psi_\theta}\mathbb{S},
\end{align}
which changes throughout the optimization.
As detailed below, we obtain PII and SR by orthogonally projecting their functional schemes onto these finite-dimensional spaces. Any solution $\xi\in\mathbb R^P$ to the linear system
\begin{align}
    \sum_j b(\partial_{\theta_i}\hat\psi_{\theta}, \partial_{\theta_j}\hat\psi_{\theta}) \, \xi_j = b(\partial_{\theta_i}\hat\psi_{\theta}, d_k) \label{eq:general_linear_system}
\end{align}
represents the orthogonal projection of $\ket{d_k}$ with respect to the inner product $b$ into the space $\mathcal V_\theta$ in the generating system $\partial_{\theta_1}\hat\psi_\theta, \dots, \partial_{\theta_1}\hat\psi_\theta$. Such a projection is called a \emph{Galerkin projection}.

To derive PII from II, we consider the inner product $b(v, w) = \mel{v}{\hat{H} - \tau \hat{I}}{w}$ and project the inverse iteration direction $d_k$ in Eq.~\eqref{eq:tangent_inverse_iter_step_non_parametric} onto $\mathcal V_\theta$ with respect to the inner product $b$. 
This choice comes from the fact that if we choose $w = d_k$, the intractable $(\hat{H} - \tau \hat{I})^{-1}$ in the functional II update in Eq.~\eqref{eq:tangent_inverse_iter_step_non_parametric} is annihilated. This will allow us to make II tractable through projection. Computing $b$ and $d_k$ in the generating system of $\mathcal V_\theta$ yields
\begin{align}
b(\partial_{\theta_i}\hat\psi_{\theta}, \partial_{\theta_j}\hat\psi_{\theta}) 
    &=
\mel{\partial_{\theta_i}\hat\psi_{\theta}}{\hat H - \tau I}{\partial_{\theta_j}\hat\psi_{\theta_j}}
    =
    Q(\theta)_{ij}
    \\
    b(\partial_{\theta_i}\hat\psi_{\theta}, d_k) 
    &=
    \mel{\partial_{\theta_i}\hat\psi_{\theta}}{\hat H - \tau I}{d_k}
    =
    \tfrac12\partial_{\theta_i}E(\psi_{\theta})
\end{align}
where we introduced the matrices
\begin{align}
    Q(\theta) &= 
    H(\theta) - \tau S(\theta) \label{eq:Q_matrix} \\
    S(\theta)_{ij} 
    &=
    \braket{ \partial_{\theta_i} \hat\psi_\theta}{ \partial_{\theta_j} \hat\psi_\theta } \label{eq:overlap_matrix} \\
    H(\theta)_{ij} 
    &= \mel{\partial_{\theta_i}\hat\psi_\theta }{\hat H}{\partial_{\theta_j} \hat\psi_\theta}\label{eq:H_matrix}
\end{align}
which can be estimated efficiently with Monte Carlo sampling, as we demonstrate in the next section.
The $S$-matrix in Eq.~\eqref{eq:overlap_matrix} is commonly referred to as the \emph{Gram matrix}, \emph{overlap matrix} or \emph{quantum geometric tensor}~\cite{stokes2020quantum}. The $H$-matrix represents the discretization of the Hamiltonian $\hat{H}$ and appears also in e.g.\ the linear method~\cite{toulouse2007optimization}.
For PII, we obtain the linear system $Q(\theta)\xi=\frac12\nabla_\theta E(\hat\psi_\theta)$, clarifying the preconditioner and parameter-update rule of PII in Eq.~\eqref{eq:preconditioned_gradient_update_PII}. 

We can also derive SR from a similar Galerkin-projection viewpoint. For SR the direction is the Riemannian gradient $\ket{d_k} = \operatorname{grad}_\mathbb SE(\hat \psi_k) = 2(\hat H - E(\hat\psi_k))\ket{\hat \psi_k}$, and we use the $L^2$ inner product instead of $b$: $b(v,w) = \braket{v}{w}$. This yields the familiar formulas
\begin{align*}
    \braket{\partial_{\theta_i}\hat\psi_{\theta}}{\partial_{\theta_j}\hat\psi_{\theta}}
    &=
    S(\theta)_{ij}
    \\
    \braket{\partial_{\theta_i}\hat\psi_{\theta}}{d_k}
    &=
    \partial_{\theta_i}E(\psi_\theta),
\end{align*}
Any solution $\xi\in\mathbb R^P$ to the linear system $S(\theta)\xi=\nabla_\theta E(\psi_\theta)$ now represents the orthogonal projection of the Riemannian gradient into $\mathcal V_\theta$ with respect to the $L^2$ inner product.
This shows that SR is in fact an $L^2$ projected Riemannian gradient descent scheme. Note that the standard $L^2$ inner product lacks information about the Hamiltonian spectrum. This lies at the core of the struggling behavior of SR in certain spectral regimes.

Finally, our Galerkin viewpoint also gives insight into the properties of the Rayleigh-Gauss-Newton method \cite{webber2022rayleigh} as a Riemannian method. The latter can now be understood as follows: $2b$ is the Riemannian Hessian of the energy, where $b(v, w) = \mel{v}{\hat{H} - E(\psi)\hat{I})}{w}$. When discretized, it becomes the matrix $H(\theta) - E(\theta)S(\theta)$. 

In our formalism, we can categorize the existing optimization schemes in terms of the corresponding function space optimizer and eigensolver, as shown in Table~\ref{tab:categorization}. Hence, to summarize the procedure to produce a parameter update rule from a given eigensolver (and therefore reproduce Table~\ref{tab:categorization}) goes as follows: (i) a chosen eigenvalue solver yields a functional update direction $\ket{d_k}$, which (ii) induces a choice for the inner product $b$, in order to make $b(v, d_k)$ tractable. Through (iii) Galerkin projections, we obtain a linear equation to define a parameter update rule for variational models.

\begin{table}
    \label{tab:algorithms}
    \centering
    \renewcommand{\arraystretch}{1.6}  
    \begin{tabularx}{\textwidth}{|C|C|C|C|}
    \hline
    {\textbf{Parameter algorithm}} & \textbf{Function space optimizer} & {\textbf{Inner product} $b(v, w)$} & {\textbf{Eigensolver}} \\
    \hline\hline
    Projected inverse iteration (this work) & Globalized Riemannian Newton method on $\mathbb{S}$ & $\mel{v}{\hat{H} - \tau \hat{I}}{w}$ & Shifted inverse iteration \\\hline
    Rayleigh Gauss-Newton \cite{webber2022rayleigh} & Riemannian Newton method on $\mathbb{S}$ & $\mel{v}{\hat{H} - E(\psi)}{w}$ & Rayleigh quotient iteration \\\hline
    Stochastic reconfiguration & Riemannian $L^2$ gradient descent on $\mathbb{S}$ & $\braket{v}{w}$ & Power iteration \\\hline
    Wasserstein quantum Monte Carlo \cite{neklyudov2023wasserstein}  & Wasserstein gradient descent & {\raisebox{-\height}{--}} & {\raisebox{-\height}{--}} \\\hline
    \end{tabularx}
    \caption{Correspondence of optimization algorithms across parameter space, function space, and ---where applicable---eigensolvers. We expand on the interpretation of SR as power iteration in Appendix~\ref{sec:sr_is_power_iter}\label{tab:categorization}. }
\end{table}

\subsection{Monte Carlo estimators}
Our PII approach is efficient due to the availability of efficient Monte Carlo estimators of the $S$ and $H$ matrix in Eqs.~\eqref{eq:overlap_matrix} and \eqref{eq:H_matrix}. We focus on unnormalized complex-valued wave functions $\psi_\theta \in \mathbb{C}$ with real parameters $\theta \in \mathbb{R}^P$, but the method generalizes to other cases. We generate a set of $M$ samples $\{x^{(m)}\}_{m=1}^{M}$, distributed according to the Born probability $p_\theta (x) = \abs{\hat{\psi}_\theta^2}$. The energy is obtained in the standard way as $E = \mcest{p}{E_L}$ with the local energy $E_L(x) = [H \psi_\theta](x) / \psi_\theta(x)$.
The estimators for the energy gradient and the matrices $S$ and $H$ read

\begin{align}
    \partial_{\theta_i} E &= 2\Re\mcest{p}{(J_i-\mcest{p}{J_i})^* (E_L - E)} \\
    &\approx (\mathbf{O}^\top \cdot \mathbf{e})_i \label{eq:MC_Egrad}\\
    S(\theta)_{ij} &= \Re\mcest{p}{ (J_i - \mcest{p}{J_i})^*  (J_j - \mcest{p}{J_j}) } \\ 
    &\approx (\mathbf{O}^\top \cdot \mathbf{O})_{ij} \label{eq:MC_S} \\
    H(\theta)_{ij} &= \Re\mcest{p}{\left(J_i - \mcest{p}{ J_i}\right)^* \left(E_{L,j} + E_L \left( J_j - \mcest{p}{ J_j}\right)\right)} \\
    &\approx (\mathbf{O}^\top \cdot \mathbf{A})_{ij} \label{eq:MC_H}
\end{align}
We introduced the log derivative $J_i(x) = \partial_{\theta_i} \log \psi_\theta(x)$, and $E_{L,i}(x) = \partial_{\theta_i} E_L(x)$ denotes the derivative of the local energy. The matrices in the final expressions read
\begin{align}
    \mathbf{e}_{n} &= \frac{2}{\sqrt{M}} \ReIm_n\left(E_L(x^{(m)}) - E\right) &\mathbf{e} \in \mathbb{R}^{2M} \label{eq:MC_e_mat} \\
    \mathbf{O}_{nj} &= \frac{1}{\sqrt{M}} \ReIm_n\left(J_j(x^{(m)}) - \mcest{p}{J_j}\right) &\mathbf{O}\in \mathbb{R}^{2M \times P} \label{eq:MC_O_mat} \\
    \mathbf{A}_{nj} &=  \frac{1}{\sqrt{M}} \ReIm_n \left(E_{L,\theta}(x^{(m)}) + E_L(x^{(m)}) (J_j(x^{(m)}) - \mcest{p}{J_j})\right) &\mathbf{A}\in \mathbb{R}^{2M \times P} \label{eq:MC_A_mat}
\end{align}
where $m = [(n-1)~\text{mod}~M] + 1$, i.e., we count $m=1,\dots,M,1,\dots,M$, and $\ReIm_n$ indicates that the real part must be considered for $n = 1, \ldots, M$ and the imaginary part for $n = M+1, \ldots, 2M$. 
While computing the Jacobian $J$ and the local energies $E_L$ are standard practices in SR as well, the local energy derivative $E_{L,i}$ is new in PII. However, for spin systems and neural network wave functions, the gradients in $E_{L,i}$ of the wave function can be computed with a computational cost scaling similarly to $E_L$. Indeed, backpropagation for many neural network architectures has the same scaling as a forward pass, typically roughly adding a constant factor $3$ in FLOPs~\cite{jax2018github}, independent of the system size.

Despite the symmetry of $H(\theta)$ in Eq.~\eqref{eq:H_matrix}, the Monte Carlo estimator in Eq.~\eqref{eq:MC_H} is only asymptotically Hermitian. In practice, we use a Tikhonov regularization for the preconditioner
\begin{align}
    Q(\theta) = H(\theta) - \tau S(\theta) + \lambda \mathbb{I}_P. \label{eq:regularization}
\end{align}
Notice that $\lambda$ has the unit of energy here, while our learning rate $\eta$ is dimensionless. Hence, the magnitudes of $\lambda$ and $\eta$ for PII and SR have no clear relation. We remark that the choice of regularization strength is in general a delicate matter \cite{feischl2024regularized, lubich2025regularized}.

We can further use the push-through identity for the inversion of the empirical preconditioner
\begin{align}
    \left(\mathbf{O}^\top  \mathbf{A} - \tau \mathbf{O}^\top \mathbf{O} + \lambda \mathbb{I}_P\right)^{-1} \mathbf{O}^\top  \mathbf{e}
    = \mathbf{O}^\top \left(\mathbf{O}  \mathbf{A}^\top - \tau \mathbf{O} \mathbf{O}^\top + \lambda \mathbb{I}_{2M}\right)^{-1}   \mathbf{e} \label{eq:minPII}
\end{align}
In the case $2M \leq P$ we now need to invert the smaller matrix
\begin{align}
    \mathbf{O}  \mathbf{A}^\top - \tau \mathbf{O} \mathbf{O}^\top + \lambda \mathbb{I}_{2M} \in \mathbb{R}^{2M \times 2M} \label{eq:Qbar}
\end{align}
We will refer to this approach as MinPII, in accordance with MinSR that uses the same push-through identity~\cite{chen2024empowering}. In the numerical experiments, we use minPII whenever $2M < P$ and PII otherwise. This yields a computational cost of $\mathcal O(\min(M^2P, P^2M))$ due to the matrix-matrix multiplications involved in forming the preconditioner. We use scipy's general linear solve to obtain the solution to the linear problem in all cases.

\subsection{Kaczmarz-inspired acceleration}
\label{sec:kaczmarz_acceleration}
Here, we discuss how to adapt the ideas of the SPRING algorithm \cite{goldshlager2024kaczmarz} (or similarly MARCH~\cite{gu2025solving}) to the setting of PII. SPRING improves the sample efficiency of SR and thereby frequently improves its empirical performance. We observe a similar behavior for PII as illustrated in Supplemental Material~\ref{sec:spring_pii}, motivating the theoretical derivation in this section.

The derivation of SPRING requires a least-squares formulation of the update direction. Because of the non-symmetric Monte Carlo estimator of the $H$-matrix in Eq.~\eqref{eq:MC_H}, a least-squares formulation of PII is not clear a priori. For flexibility, we specialize the viewpoint of Ref.~\cite{mullerposition} to the VMC setting. The linear system in the Galerkin projection in Eq.~\eqref{eq:general_linear_system} can equivalently be solved as a minimization problem
\begin{align}
    \xi^* &= \underset{\xi \in \mathbb{R}^P}{\operatorname{argmin}} \, \mathcal{L}(\xi) , \quad\text{where }    \mathcal{L}(\xi) =  \frac{1}{2}\norm{\sum_{j=1}^P \ket{\partial_j \hat{\psi}_k} \xi_j - \ket{d_k}}_b^2 .
\end{align}
An important consequence of our framework is that we defined the norm to be with respect to a general inner product $b$. 
By expanding the objective above and using the definitions of the preconditioner $Q(\theta_k)$ and the parameter gradient $\nabla_\theta E(\psi_k)$, the minimization objective reduces to the quadratic form:
\begin{align}
    \mathcal{L}(\xi) 
    &=
    \frac{1}{2} \xi^\top Q(\theta_k) \xi - \frac12 \xi^\top \nabla_\theta E(\psi_k) + \text{const}.
\end{align}
From this perspective, adding an $L_2$ penalty $\lambda \norm{\xi}_2^2$ to $\mathcal{L}$ to keep the norm of $\xi$ small recovers the standard Tikhonov-regularized linear system in Eq.~\ref{eq:regularization}. 
This equivalence is straightforward to verify for SR, where $b$ simply reduces to the standard $L_2$ norm and $Q$ reduces to the overlap matrix $S$~\cite{stokes2020quantum}. 

To achieve the sample efficiency of SPRING through Kaczmarz within PII, we replace the origin-centered penalty above with a term anchored to the update direction from the previous VMC step, $\xi_{k-1}$
\begin{align}
    \xi^* 
    &=
    \underset{\xi \in \mathbb{R}^P}{\operatorname{argmin}} \, \left[\mathcal{L}(\xi) + \frac{\lambda}{2} \norm{\xi - \mu \xi_{k-1}}_2^2\right] \\
    &= \underset{\xi \in \mathbb{R}^P}{\operatorname{argmin}} \, \left[ \frac{1}{2} \xi^T Q(\theta_k) \xi - \frac12\xi^T \nabla_\theta E(\psi_k) +  \mu \xi^T Q(\theta_k) \xi_{k-1} + \frac{\lambda}{2} \xi^T \xi \right] + \mu \xi_{k-1} 
\end{align}
Converting the above back into a linear problem yields PII-SPRING
\begin{align}
\left(Q(\theta_k) + \lambda \mathbb{I}_P\right) (\xi - \mu \xi_{k-1}) &= \frac12\nabla_\theta E(\psi_k) - \mu Q(\theta_k) \xi_{k-1}
\end{align}
In terms of Monte Carlo estimators, we obtain the solution
\begin{align}
    \xi^* = \mu\xi_{k-1} + \left(\mathbf{O}^T \mathbf{A} - \tau \mathbf{O}^T \mathbf{O} + \lambda \mathbb{I}_{2M}\right)^{-1} \mathbf{O}^T \left[\tfrac12\mathbf{e} - \mu \left( \mathbf{A} - \tau \mathbf{O}\right) \xi_{k-1} \right]
\end{align}
which also has a straightforward MinPII-SPRING form using the push-through identity as in Eq.~\eqref{eq:minPII}.

\subsection{Convergence results}\label{sec:proofs}
We present the precise convergence results for PII and SR in the non-parametric setting, where these algorithms correspond to inverse iteration and Riemannian gradient descent with retraction by normalization, respectively. Further details on inverse iteration and eigenvalue methods can be found in Ref.~\cite{golub2013matrix}, and on Riemannian optimization in Refs.~\cite{absil2008optimization, boumal2023introduction}. In the following, we consider a Hermitian Hamiltonian $\hat{H} \in \mathbb{C}^{2^{n}\times 2^{n}}$ acting on the Hilbert space $\mathbb{H} = \mathbb{C}^{2^n}$. We denote its distinct eigenvalues, which may correspond to degenerate eigenspaces, by $E_0 < E_1 < \dots < E_{\textup{max}}$. Letting $V_0$ represent the ground state eigenspace associated with $E_0$, we define the system's spectral gap as $\Delta = E_1 - E_0 > 0$, and the system's spectral spread as $\Gamma = E_{\text{max}} - E_0$. By $P_0$ and $P_0^\perp$ we denote the orthogonal projection onto $V_0$ and its orthogonal complement, respectively.

\begin{restatable}[Convergence of Inverse Iteration]{theorem}{theoreminverseiter}
    \label{theorem:convergence_shifted_inverse_iteration}
    Let the sequence $(\psi_k)_{k\in\mathbb N}$ be generated by inverse iteration with shift $\tau < \frac12(E_0 + E_1)$, i.e.,
    \[
        \psi_{k+1}
        =
        \frac{(\hat H - \tau I)^{-1}\psi_k}{\| (\hat H - \tau I)^{-1}\psi_k \|}.
    \]
    Assume that the initial iterate $\psi_0 \in \mathbb S$ has nonzero overlap with $V_0$, meaning $P_0\psi_0$ is non-zero and set $\psi^*\coloneqq P_0\psi_0/\lVert P_0\psi_0 \rVert\in V_0$. Then, it holds
    \begin{equation}\label{eq:rate_inverse_iter_rigorous}
        \|\psi_k \pm \psi^* \| 
        \leq 
        \left( \sqrt 2 \frac{\|P_0^\perp\psi_0\|}{\|P_0\psi_0\|}\right)
            \left|
                1 - \frac{E_1 - E_0}{E_1 - \tau}
            \right|^k. 
    \end{equation}
    The sign in front of $\psi^*$ in the equation above is given by $\operatorname{sign}(E_0 - \tau)$.
\end{restatable}

\vspace{-2em}

\begin{restatable}[Convergence of Riemannian Gradient Descent]{theorem}{theoremriemanniangradient}
    \label{thm:convergence_RG}
    Let the sequence $(\psi_k)_{k\in\mathbb N}$ be generated by Riemannian gradient descent with retraction by normalization, i.e.,
    \begin{align*}
        \psi_{k+1} 
        =
        \frac{\psi_k - \eta\operatorname{grad}_\mathbb SE(\psi_k)}{\| \psi_k - \eta\operatorname{grad}_\mathbb SE(\psi_k) \|},
        \quad
        \operatorname{grad}_\mathbb S E(\psi_k) 
        =
        2[\hat H\psi_k - E(\psi_k)\psi_k].
    \end{align*}
    Assume that the initial iterate $\psi_0 \in \mathbb S$ has nonzero overlap with $V_0$, meaning $P_0\psi_0$ is non-zero and set $\psi^*\coloneqq P_0\psi_0/\lVert P_0\psi_0 \rVert\in V_0$. Then, choosing the stepsize $\eta = 1/(2(E_{\textup{max}} - E_0))$ it holds
    \begin{equation*}
        \| \psi_k - \psi^* \|
        \leq
        \left( \sqrt2 \frac{\|P^\perp_0\psi_0\|}{\| P_0\psi_0 \|} \right)
        \left | 1 - \frac{1}{2}\frac{E_1 - E_0}{E_{\textup{max}} - E_0} \right|^k
    \end{equation*}
    Moreover, the algorithm converges for every step-size $\eta < 1/(E_{\textup{max}} - E_0)$.
\end{restatable}

The proofs for both results are provided in Appendix~\ref{sec:proofs}. As discussed in the Section~\ref{sec:methods}, these theorems establish the global linear convergence of both algorithms. They illustrate how the convergence rate of Riemannian gradient descent deteriorates as the spectral gap closes and the spread widens, while inverse iteration can achieve a gap-independent and arbitrarily fast rate as $\tau\to E_0$. Notably, the explicit constants in these bounds are identical for Riemannian gradient descent and inverse iteration. While the general convergence properties of these underlying methods are well-known in principle, we are not aware of existing literature that provides a complete, global non-asymptotic analysis with fully explicit constants that also accommodates degenerate ground states. For this reason, we state the rigorous results in full here and provide the complete proofs in Appendix~\ref{sec:proofs}. We emphasize that our core theoretical contribution is the realization that the analysis above directly describes SR and PII, as these are the Galerkin-projected variants of Riemannian gradient descent and inverse iteration.

\section{Conclusions and Outlook}
We introduced Projected Inverse Iteration (PII) to optimize neural network wavefunctions by framing the ground-state search as an eigenvalue problem rather than an energy minimization task. Our theoretical framework identifies Stochastic Reconfiguration (SR) as a Galerkin-projected Riemannian gradient descent, which explains its sensitivity to closing spectral gaps. In contrast, PII projects classical inverse iteration onto the tangent space of the variational manifold. Numerical experiments on highly frustrated two-dimensional spin systems confirm that PII is robust to closing spectral gaps, yielding significantly faster convergence than SR while maintaining the same favorable computational scaling.

Future research will focus on extending PII to target excited states. Additionally, integrating PII with variance reduction techniques and advanced sampling methods will be necessary to adapt the algorithm for complex fermionic systems in quantum chemistry and condensed matter physics.

Beyond quantum mechanics, PII demonstrates the utility of an infinite-dimensional perspective on optimization in scientific machine learning, as suggested in \cite{mullerposition}. Because PII can be interpreted as a natural gradient method tailored specifically for eigenvalue problems, it maps directly to general non-linearly parameterized eigenvalue problems. We anticipate this makes PII a robust, scalable tool for a wider class of challenges in computational science.

\backmatter

\bmhead{Acknowledgements}
The simulations in this work were based on
Netket~\cite{vicentini2022netket, carleo2019netket} and Jax~\cite{jax2018github}.

\section*{Declarations}
MZ acknowledges support from an ETH Postdoctoral Fellowship for the project ``Reliable, Efficient, and Scalable Methods for Scientific Machine Learning''. This work was supported by a grant from the Swiss National Supercomputing Centre (CSCS) under project ID lp20 on Alps.

\section*{Code availability} An implementation of the algorithms presented in this work are available on the github: \url{https://github.com/jwnys/projected_inverse_iteration}.

\begin{appendices}

\section{Variational model: ViT recap}
The variational state used in this work is based on the previously introduced ViT architecture for lattice spin systems~\cite{rende2024simple}. A spin configuration on an $L\times L$ square lattice is divided into $b\times b$ blocks, yielding $N=L^2/b^2$ patches. Each patch is mapped linearly to a $d$-dimensional embedding, defining the input sequence $X^{(0)}=[\boldsymbol{x}_1,\dots,\boldsymbol{x}_N]$, with $\boldsymbol{x}_i\in\mathbb{R}^d$. The sequence is then processed by $n_\ell$ stacked encoder blocks. Each block contains a multi-head factored attention (MHFA) layer and a point-wise feed-forward network of hidden size $2d$ with ReLU activation, with residual connections and pre-layer normalization throughout. Given an input sequence $X^{(\ell)}$, the MHFA layer produces an intermediate sequence $A^{(\ell)}=[\boldsymbol{A}_1,\dots,\boldsymbol{A}_N]$, where $\boldsymbol{A}_i\in\mathbb{R}^d$, according to
\begin{align}
A_{i,p}
=
\sum_{q=1}^{d}W_{p,q}\sum_{j=1}^{N}\alpha_{ij}^{\mu(q)}\sum_{r=1}^{d}V_{q,r}x^{(\ell)}_{j,r}.
\end{align}

Here, $x^{(\ell)}_{j,r}$ denotes the $r$-th component of the $d$-dimensional representation of patch $j$ at layer $\ell$; $\mu(q)=\lceil q h / d\rceil$ maps the feature index $q$ to its corresponding attention head; $V \in \mathbb{R}^{d \times d}$ and $W \in \mathbb{R}^{d \times d}$ are learnable weight matrices; and $\alpha^\mu \in \mathbb{R}^{N \times N}$ is the attention matrix.
Unlike standard dot-product attention, where the coefficients are generated from query--key overlaps through a softmax, MHFA\cite{rende2024simple} treats the attention weights themselves as variational parameters. On the periodic patch lattice, these weights are taken in relative form $\alpha_{ij}^{\mu}=\alpha_{i-j}^{\mu},$
with the index difference understood modulo the patch lattice. This parameterization is compatible with lattice translations and reduces the number of attention parameters from $O(N^2)$ to $O(N)$.
The encoder block then updates the sequence $X^{(\ell)}$ according to
\begin{align}
\tilde{X}^{(\ell)} &= X^{(\ell)}+\mathrm{MHFA}\!\left(\mathrm{LayerNorm}(X^{(\ell)})\right), \\
X^{(\ell+1)} &= \tilde{X}^{(\ell)}+\mathrm{MLP}\!\left(\mathrm{LayerNorm}(\tilde{X}^{(\ell)})\right).
\end{align}
After $n_\ell$ blocks, the final sequence $X^{(n_\ell)}=(\boldsymbol{y}_1,\dots,\boldsymbol{y}_N)$ is summed over patches $N$ to give a single $d$-dimensional representation, $\boldsymbol{z}=\sum_{i=1}^{N}\boldsymbol{y}_i \in \mathbb{R}^d$. 
The variational wavefunction is then parameterized as $\log \Psi_\theta(\sigma)=f_\theta(\sigma)+ i\, g_\theta(\sigma),$, with 
\begin{align}
f_\theta(\sigma)&=\sum_{\alpha=1}^{d}\phi\!\left(b_\alpha^{(f)}+\boldsymbol{w}_\alpha^{(f)}\!\cdot\!\boldsymbol{z}\right),\\
g_\theta(\sigma)&=\sum_{\alpha=1}^{d}\phi\!\left(b_\alpha^{(g)}+\boldsymbol{w}_\alpha^{(g)}\!\cdot\!\boldsymbol{z}\right),
\end{align}
where $\phi(x)=\log\cosh x$. All parameters $b$ and $w$ are real valued. For stoquastic Hamiltonians, a positive real-valued ansatz is used, obtained by omitting the imaginary part $g_\theta(\sigma)$. For frustrated systems the complex-valued form is used.

\section{Practical implementation details}

In this section, we present the pseudocode for the parameter-space implementation of PII used in the Results section (Algorithm~\ref{alg:pii_paramspace}), based on the Monte Carlo estimators introduced in Eqs.~\eqref{eq:MC_Egrad}--\eqref{eq:MC_H}. At iteration $k$, samples are drawn from the current variational state to estimate the projected quantities, after which the regularized linear system in Eq.~\eqref{eq:regularization} is solved to obtain the update direction in parameter space. 
\begin{algorithm}[tbh]
\centering
\begin{small}
\begin{algorithmic}[1]
  \Require regularization $\lambda$, learning-rate schedule $\{\eta_k\}_{k=0}^{K-1}$, number of iterations $K$, sample count $M$, initial parameters $\theta_0$, energy shift $\tau$
  \For{$k = 0, \dots, K-1$}
    \State Draw $M$ MCMC samples
    \State Estimate $\mathbf{e}_k$, $\mathbf{O}_k$, and $\mathbf{A}_k$ \qquad (see Eqs.~\eqref{eq:MC_e_mat}--\eqref{eq:MC_A_mat})
    \State Form $\nabla_\theta E(\psi_k) \gets \mathbf{O}_k^\top \mathbf{e}_k$
    \State Form $Q_k \gets \mathbf{O}_k^\top \mathbf{A}_k - \tau \mathbf{O}_k^\top \mathbf{O}_k + \lambda \mathbb{I}$ \qquad (see Eqs.~\eqref{eq:MC_S}--\eqref{eq:MC_H} \& Eq.~\eqref{eq:regularization} for PII, or Eq.~\eqref{eq:Qbar} for MinPII)
    \State Solve the linear system $\xi_k = Q_k^{-1} [\tfrac12\nabla_\theta E(\psi_k)]$
    \State Update parameters: $\theta_{k+1} \gets \theta_k - \eta_k \xi_k$ \qquad (see Eq.~\eqref{eq:preconditioned_gradient_update_PII})
  \EndFor
  \State \Return $\theta_K$
\end{algorithmic}
\end{small}
\caption{Practical implementation of PII}
\label{alg:pii_paramspace}
\end{algorithm}

\section{Numerical details: hyperparameters for the experiments}

Table~\ref{tab:training_hparams} shows the architecture and optimization hyperparameters used in the numerical experiments reported in the Results section. For the ViT ansatz, the architectural parameters are the number of transformer layers $n_\ell$, the number of attention heads $h$, the embedding dimension $d$, and the patch size $b$. For the RBM ansatz, we report the hidden-unit density $\alpha$. The optimization parameters are the number of Monte Carlo samples $M$, the learning rate $\eta$, the regularization parameter $\lambda$, and, where applicable, the energy shift $\tau$. 

\begin{table*}[htb]
\centering
\small
\setlength{\tabcolsep}{4pt}
\caption{Architecture and optimization hyperparameters used in the numerical experiments.}
\label{tab:training_hparams}
\resizebox{\textwidth}{!}{%
\begin{tabular}{@{}lllllccccc cccc@{}}
\toprule
\multirow{2}{*}{Figure} & \multirow{2}{*}{System} & \multirow{2}{*}{Size} & \multirow{2}{*}{Method} & \multirow{2}{*}{Ansatz} & \multicolumn{5}{c}{Architecture} & \multicolumn{4}{c}{Optimization} \\
\cmidrule(lr){6-10}\cmidrule(lr){11-14}
& & & & & $n_\ell$ & $h$ & $d$ & $b$ & $\alpha$ & $M$ & $\eta$ & $\lambda$ & $\tau$ \\
\midrule

Fig.~\ref{fig:spectral_gaps} & TFIM & $4\times4$ & PII & real ViT & 4 & 2 & 20 & 2 & \textemdash & 32768 & 1.0 & $10^{-2}$ & $E_0$ \\
Fig.~\ref{fig:spectral_gaps} & TFIM & $4\times4$ & SR  & real ViT & 4 & 2 & 20 & 2 & \textemdash & 32768 & 0.01 & $10^{-4}$ & \textemdash \\
Fig.~\ref{fig:spectral_gaps} & TFIM & $4\times4$ & Adam  & real ViT & 4 & 2 & 20 & 2 & \textemdash & 32768 & 0.01 & $10^{-4}$ & \textemdash \\
Fig.~\ref{fig:undershoot_batch_sizes}a & TFIM & $4\times4$ & PII & RBM & \textemdash & \textemdash & \textemdash & \textemdash & 4 & 512, 4096, 16384 & 1.0 & $10^{-1}$ & $E_0$ \\
Fig.~\ref{fig:undershoot_batch_sizes}b & TFIM & $4\times4$ & PII & RBM & \textemdash & \textemdash & \textemdash & \textemdash & 4 & 512, 4096, 16384 & 1.0 & $10^{-1}$ & $1.2\,E_0$ \\
Fig.~\ref{fig:undershoot_batch_sizes}c & TFIM & $4\times4$ & PII & RBM & \textemdash & \textemdash & \textemdash & \textemdash & 4 & full sum & 1.0 & $10^{-1}$ & $\alpha E_0$ \\
Fig.~\ref{fig:10x10_systems}a & TFIM & $10\times10$ & PII & real ViT & 4 & 2 & 20 & 2 & \textemdash & 32768 & 1.0 & $10^{-3}$ & 1.1$E_0$ \\
Fig.~\ref{fig:10x10_systems}a & TFIM & $10\times10$ & SR  & real ViT & 4 & 2 & 20 & 2 & \textemdash & 32768 & 0.002 & $10^{-4}$ & \textemdash \\
Fig.~\ref{fig:10x10_systems}b & Heisenberg & $10\times10$ & PII & real ViT & 4 & 2 & 20 & 2 & \textemdash & 32768 & 1.0 & $10^{-2}$ & 1.4$E_0$ \\
Fig.~\ref{fig:10x10_systems}b & Heisenberg & $10\times10$ & SR  & real ViT & 4 & 2 & 20 & 2 & \textemdash & 32768 & 0.002 & $10^{-4}$ & \textemdash \\
Fig.~\ref{fig:10x10_systems}c & $J_1$--$J_2$ & $10\times10$ & MinPII & complex ViT & 4 & 10 & 60 & 2 & \textemdash & 16384 & 1.0 & $10^{-1}$ & $E_0$ \\
Fig.~\ref{fig:10x10_systems}c & $J_1$--$J_2$ & $10\times10$ & MinSR  & complex ViT & 4 & 10 & 60 & 2 & \textemdash & 16384 & 0.004 & $10^{-4}$ & \textemdash \\
Fig.~\ref{fig:high_lr_failure}a & TFIM & $10\times10$ & SR & real ViT & 4 & 2 & 20 & 2 & \textemdash & 32768 & $0.001 \rightarrow 0.01$ & $10^{-3}$ & \textemdash \\
Fig.~\ref{fig:high_lr_failure}b & TFIM & $10\times10$ & SR & real ViT  & 4 & 2 & 20 & 2 & \textemdash & 32768 & $0.001 \rightarrow 0.01$ & $10^{-4}$ & \textemdash \\
Fig.~\ref{fig:RQI_J1_H_matrix}a & TFIM & $4\times4$ & PII & RBM & \textemdash & \textemdash & \textemdash & \textemdash & 4 & full sum & 1.0 & $10^{-1}$ & $1.4E_0$ \\
Fig.~\ref{fig:RQI_J1_H_matrix}b & TFIM & $4\times4$ & PII & RBM & \textemdash & \textemdash & \textemdash & \textemdash & 4 & 4096 & 1.0 & $10^{-1}$ & $1.4E_0$ \\
Fig.~\ref{fig:pii_spring} & TFIM & $10\times10$ & PII & real ViT & 4 & 2 & 20 & 2 & \textemdash & 4096, 32768 & 1.0 & $10^{-3}$ & 1.1$E_0$ \\
Fig.~\ref{fig:pii_spring} & TFIM & $10\times10$ & PII-SPRING & real ViT & 4 & 2 & 20 & 2 & \textemdash & 4096, 32768 & 1.0 & $10^{-3}$ & 1.1$E_0$ \\

\bottomrule
\end{tabular}}
\end{table*}

\section{Details about finetuning learning rate of SR: demonstration of failure at higher LR}

\begin{figure}[t]
  \centering
  \includegraphics[width=\textwidth]{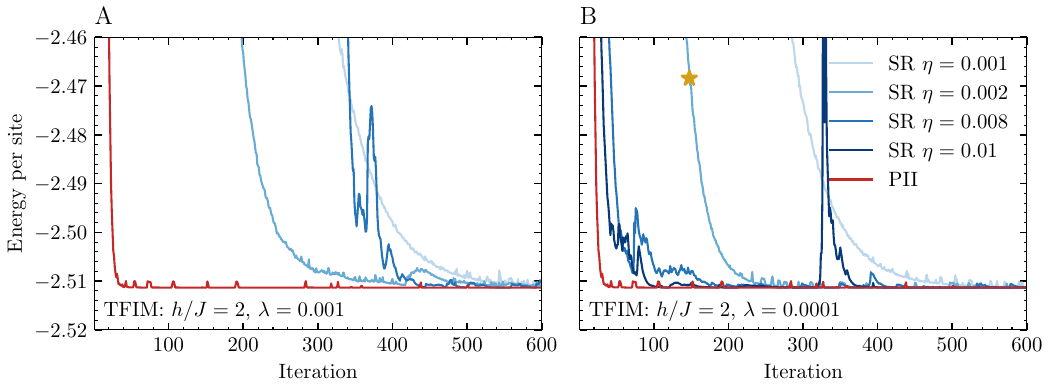}
\caption{\textbf{SR is sensitive to learning-rate tuning on the $10\times10$ TFIM.}
Energy per site versus iteration for TFIM at $h/J=2$, comparing SR across several learning rates with PII.
\textbf{A,} Results for $\lambda=10^{-3}$.
\textbf{B,} Results for $\lambda=10^{-4}$.
As the SR learning rate is increased, the optimization becomes progressively less stable, with transient spikes appearing in the estimated ground state energy. By contrast, PII converges to the ground state more rapidly than all SR settings shown here. The yellow star marks the SR hyperparameter setting used in the main text, corresponding to the best SR performance considered here.}
  \label{fig:high_lr_failure}
\end{figure}

In this section, we show that SR was tuned as carefully as possible to obtain its strongest performance. Even under such tuning, SR remains sensitive to the joint choice of learning rate and diagonal shift. Figure~\ref{fig:high_lr_failure}a,b shows the optimization trajectories obtained with two diagonal shifts, $\lambda=10^{-3}$ and $\lambda=10^{-4}$, respectively. The diagonal shift is introduced to improve numerical stability; however, if it is set too large, the SR convergence slows down. We therefore use $\lambda=10^{-4}$ as the best-performing SR setting. Keeping the diagonal shift fixed at $\lambda=10^{-4}$, we then increase the learning rate from $\eta=0.001$ to $\eta=0.01$. As shown in Fig.~\ref{fig:high_lr_failure}, larger learning rates progressively destabilize SR, leading to pronounced transient spikes and oscillations in the energy trajectory.

\section{Properties of the $H$ matrix}
\label{sec:H_matrix_spectrum}

To clarify the meaning of the matrices $H(\theta)$ and $S(\theta)$ defined in Eqs.~\eqref{eq:H_matrix} and \eqref{eq:overlap_matrix}, we interpret them as the effective Hamiltonian and the induced Hilbert-space metric on the variational tangent space $\mathcal V_\theta=\operatorname{span}\{\partial_{\theta_i}\hat\psi_\theta\}_{i=1}^P \subset T_{\hat\psi_\theta}\mathbb S .$
Consider an arbitrary tangent vector $\ket{\varphi}=\sum_{i=1}^P v_i\,\partial_{\theta_i}\ket{\hat\psi_\theta}\in \mathcal V_\theta ,$
with coefficient vector $v\in\mathbb R^P$. By construction, $H(\theta)$ and $S(\theta)$ reproduce the energy and norm of $\ket{\varphi}$ within this local linear space:
\begin{equation}
    \frac{v^\top H(\theta)v}{v^\top S(\theta)v}
    =
    \frac{\mel{\varphi}{\hat H}{\varphi}}{\braket{\varphi}{\varphi}}.
    \label{eq:projected_rayleigh}
\end{equation}
The generalized eigenvalue problem
\begin{equation}
    H(\theta) v = \lambda S(\theta) v
    \label{eq:generalized_H_eig}
\end{equation}
therefore coincides with the Schrödinger problem projected onto $\mathcal V_\theta$. Equivalently, the generalized eigenvalues of $(H(\theta),S(\theta))$ are the Ritz values of $\hat H$ in the tangent space. Ordering them as $\lambda_{\min}(\theta)=\lambda_1(\theta)\le\lambda_2(\theta)\le\cdots$, the smallest one satisfies
\begin{equation}
    \lambda_{\min}(\theta)
    =
    \min_{0\neq \varphi\in\mathcal V_\theta}
    \frac{\mel{\varphi}{\hat H}{\varphi}}{\braket{\varphi}{\varphi}},
    \label{eq:lambda_min_rayleigh}
\end{equation}
and thus gives the lowest Ritz value in the tangent space around $\hat\psi_\theta$. At convergence, $\hat\psi_{\theta_*}=\psi_0$, the tangent space is orthogonal to the ground state and therefore lies entirely in the excited-state sector. Accordingly, any nonzero $\varphi\in\mathcal V_{\theta_*}$ can be expanded as $\ket{\varphi}=\sum_{n\ge1} c_n\ket{E_n}$, and Eq.~\eqref{eq:lambda_min_rayleigh} reduces to
\begin{equation}
    \lambda_{\min}(\theta_*)
    =
    \min_{0\neq\varphi\in\mathcal V_{\theta_*}}
    \frac{\sum_{n\ge1}|c_n|^2 E_n}{\sum_{n\ge1}|c_n|^2}
    \ge E_1,
\end{equation}
with $E_1$ the first excited-state energy. The bound is saturated if the tangent space contains a nonzero vector in the first-excited eigenspace.

Figure~\ref{fig:RQI_J1_H_matrix} shows this behavior directly. In the full-sum calculation, the variational energy decreases toward $E_0$, whereas $\lambda_{\min}$ rapidly approaches $E_1$ and then remains near it. Once the ground-state direction is excluded from the tangent space, the lowest mode of the projected problem is controlled by the lowest excited-state direction contained in that space. 

\begin{figure}[t]
    \centering
    \includegraphics[width=0.7\textwidth]{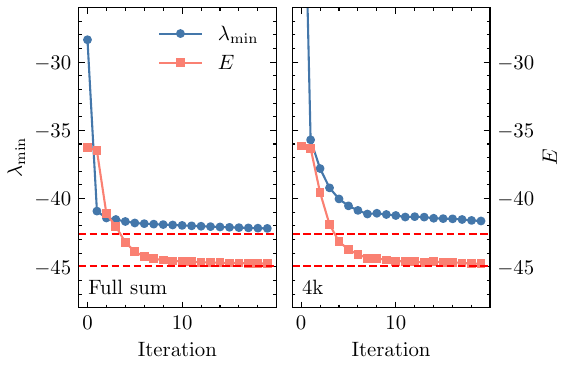}
    \caption{\textbf{Projected tangent-space spectrum during optimization on the $4\times4$ Heisenberg model.}
    Variational energy $E$ and smallest generalized eigenvalue $\lambda_{\min}$ of the matrix pair $(H(\theta),S(\theta))$ during optimization with a real-valued RBM ansatz of hidden-unit density $\alpha=4$. All runs use the same optimization hyperparameters, $\tau=1.4E_0$, $\lambda=0.1$, and $\eta=1.0$.
    \textbf{a}, Full-sum estimator.
    \textbf{b,} Monte Carlo estimator with $4$k samples.}
    \label{fig:RQI_J1_H_matrix}
\end{figure}

\section{Stochastic Reconfiguration as Power Iteration}\label{sec:sr_is_power_iter}
Here, we recall the connection between power iteration and Riemannian gradient descent with retraction by normalization. The functional update for Riemannian gradient descent can be re-written as
\begin{align*}
    \psi_{k+1}
    &=
    \frac{\psi_k - 2\eta(\hat H\psi_k - E(\psi_k)\psi_k)}{\| \psi_k - 2\eta(\hat H\psi_k - E(\psi_k)\psi_k) \|}
    \\
    &=
    \frac{(\tau_k I - \hat H)\psi_k}{\|(\tau_k I - \hat H)\psi_k\|},
\end{align*}
where $\tau_k = E(\psi_k) + \frac{1}{2\eta}$. The latter is one step of power iteration for the operator $\tau_k - \hat H$, which for a suitable choice of $\eta$ shifts the Hamiltonian such that we can compute the ground state with a power iteration. 

\section{PII-SPRING and MinPII-SPRING}\label{sec:spring_pii}

\begin{figure}[t]
  \centering
  \includegraphics[width=0.6\textwidth]{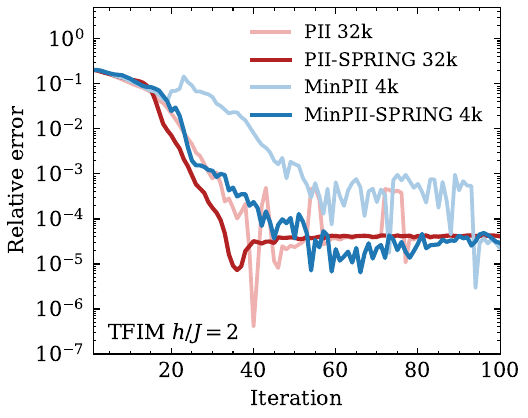}
  \caption{\textbf{SPRING accelerates convergence on the $10\times10$ TFIM for small sample size.}
  Relative error as a function of iteration for PII and PII-SPRING (MinPII and MinPII-SPRING) on the $10\times10$ TFIM at $h/J=2$ in a real-valued ViT ansatz. All runs use the same optimization hyperparameters, $\tau=1.1E_0$, $\lambda=10^{-3}$, and $\eta=1.0$. For SPRING, we used $\mu=0.8$. }
  \label{fig:pii_spring}
\end{figure}

In this section, we study whether the PII-SPRING algorithm improves the performance of PII. Fig.~\ref{fig:pii_spring} shows that the advantage of PII-SPRING is obvious in the small sample size regime. For $M=32$k, PII-SPRING exhibits a slightly faster convergence rate than PII, while the two methods become comparable at later iterations. The difference is much more pronounced for $M=4$k, where MinPII-SPRING reaches relative errors of order $10^{-4}$ substantially earlier than MinPII. This behaviour is expected when the number of samples is smaller than the number of variational parameters. The sampled linear system is then underdetermined and gives only a rough approximation to the full projected update. Updates based on the current samples alone are then more sensitive to finite-sampling noise. SPRING partly compensates for this by reusing information from last step. When the number of samples is much larger than the number of variational parameters, the sampled least-squares problem is overdetermined and typically close to uniquely determined, leaving less room for SPRING to improve the update direction.

\section{Inverse Iteration as a Riemannian Newton Method}\label{inverse_iter_as_newton}
The Riemannian Newton method with retraction by normalization is given by the recursion
\begin{equation}\label{eq:riemannian_newton}
    \psi_{k+1} 
    =
    \frac
    {
        \psi_k - \operatorname{Hess}_{\mathbb{S}} E(\psi_k)^{-1} [\operatorname{grad}_{\mathbb{S}} E(\psi_k)]
    }
    {
        \|\psi_k - \operatorname{Hess}_{\mathbb{S}} E(\psi_k)^{-1} [\operatorname{grad}_{\mathbb{S}} E(\psi_k)]\|
    }.
\end{equation}
Here, for tangent vectors $v,w\in T_\psi\mathbb S$ the Riemannian Hessian and gradient are given by
\begin{align*}
    \operatorname{Hess}_{\mathbb{S}} E(\psi)(v,w) &= 2\big[ \langle (\hat H - E(\psi)I)v, w \rangle \big] \\
    \operatorname{grad}_{\mathbb{S}} E(\psi) &= 2\big[ \hat H\psi - E(\psi)\psi \big]
\end{align*}
If $E(\psi) \notin \sigma(\hat H)$, the solution $d \in T_\psi \mathbb{S}$ to $\operatorname{Hess} E(\psi)(d, \cdot) = \langle \operatorname{grad} E(\psi), \cdot \rangle$ is given by
\[
    d = \psi - \frac{(\hat H - E(\psi)I)^{-1} \psi}{\langle \psi, (\hat H - E(\psi)I)^{-1} \psi \rangle}.
\]
This can be seen by checking that $d$ is a tangent vector, i.e., it satisfies $\langle d, \psi \rangle = 1 - 1 = 0$ and that $d$ in fact solves the correct equation $\langle (\hat H - E(\psi)I)d, w \rangle = \langle (\hat H - E(\psi)I)\psi, w \rangle$.
By using the explicit formula for $d$ in the Riemannian Newton method above, we find
\[
    \psi_{k+1} 
    =
    \frac{\psi_k - d_k}{\| \psi_k - d_k \|} 
    =
    \frac{(\hat H - E(\psi_k)I)^{-1} \psi_k}{\| (\hat H - E(\psi_k)I)^{-1} \psi_k \|}
\]
This is iterative scheme is known as Rayleigh Quotient Iteration \cite{trefethen2022numerical}. If we modify the Riemannian Hessian to be $\hat H - \tau I$ (acting only on tangent vectors) the same derivation as above shows that we obtain shifted Inverse Iteration with shift $\tau$
\begin{equation}\label{eq:inverse_iteration}
    \psi_{k+1} 
    =
    \frac{(\hat H - \tau I)^{-1} \psi_k}{\| (\hat H - \tau I)^{-1} \psi_k \|}
\end{equation}
which is well-suited when a priori information on $E_0$ is at hand, in which case we set $\tau \approx E_0$.

\section{Spectral gap analysis}
\begin{figure}[t]
  \centering
  \includegraphics[width=0.6\textwidth]{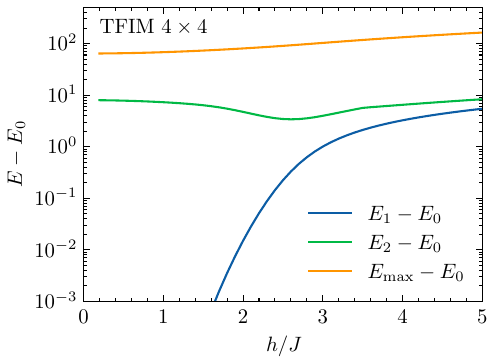}
  \caption{\textbf{Spectral gap and spread of the $4\times4$ TFIM versus $h/J$ obtained with ED.}  }
  \label{fig:energy_differences}
\end{figure}

To understand the different convergence behavior of SR and PII, we compute the spectrum of the $4\times4$ TFIM as a function of $h/J$ using ED. We focus on the first gap $E_1-E_0$, the second gap $E_2-E_0$, and the spectral width $E_{\max}-E_0$.
Figure~\ref{fig:energy_differences} clarifies the spectral regime relevant for SR optimization. At small $h / J$, the first gap $\Delta =E_1-E_0$ is strongly suppressed and falls below $10^{-3}$, making the ground state and first excited state effectively degenerate to VMC. Consequently, the optimizations are not governed solely by the lowest excitation gap; the second excitation gap $E_2-E_0$ dominates the convergence. This motivates our choice of $h / J=2$ for the comparison. Over the same range, $E_2-E_0$ remains much larger and reaches its minimum only around $h / J \approx 2.6$, whereas the spectral spread $E_{\max }-E_0$ stays large.

\section{Details of Mathematical Derivations}
In this section, we provide the proofs of the convergence results for inverse iteration and Riemannian gradient descent stated in the main text.

\theoreminverseiter*

\begin{proof}
    Denote the shifted inverse by $T_\tau = (\hat H - \tau I)^{-1}$. The assumption $\tau < \frac12(E_0 + E_1)$ guarantees that the spectrum of $T_\tau$ attains its maximum in $\mu_0 = (E_0 - \tau)^{-1}$.
    Using this notation, we write $\psi_k = c_k T_\tau\psi_0$ with $c_k = \|T^k_\tau\|^{-1}$. Moreover, we split $\psi_0 = P_0\psi_0 + P_0^\perp\psi_0 = v_0 + w_0$. Then it holds
    \begin{equation*}
        \psi_k = c_k[\mu_0^k + T_\tau^k w_0].
    \end{equation*}
    We use this to compute
    \begin{align*}
        \left\| 
            \psi_k - \operatorname{sign}(\mu_0)^k \frac{v_0}{\|v_0\|}
        \right\|^2
        &
        =
        \left\|
            c_k \mu_0^k v_0 + c_k T^k_\tau w_0 - \operatorname{sign}(\mu_0)^k \frac{v_0}{\|v_0\|}
        \right\|^2
        \\
        &=
        \left| 
            c_k |\mu_0|^k\|v_0\| - 1
        \right|^2
        +
        c_k^2 \| T_\tau^k w_0 \|^2.
    \end{align*}
    We estimate both terms above separately and start with the second. To this end, we note that it holds
    \begin{equation*}
        c_k^2
        =
        \frac{1}{\| T_\tau^k\psi_0 \|^2}
        =
        \frac{1}{\| T_\tau^k v_0 \|^2 + \|T_\tau^k w_0\|^2}
        \leq 
        \frac{1}{|\mu_0|^{2k}\|v_0\|^2}.
    \end{equation*}
    This then yields
    \begin{equation*}
        c_k^2 \| T_\tau^k w_0 \|^2
        \leq
        \frac{|\mu_1|^{2k}\|w_0\|^2}{|\mu_0|^{2k}\|v_0\|^2}.
    \end{equation*}
    To treat the second term, we need the following equality
    \begin{equation*}
        \| T_\tau^k \psi_0 \|
        =
        \sqrt{|\mu_0|^{2k}\|v_0\| + \| T_\tau^k w_0 \|^2}.
    \end{equation*}
    Using this, we compute
    \begin{align*}
        \left(
            \frac{|\mu_0|^k \|v_0\|}{\|T_\tau^k\psi_0\|} - 1
        \right)^2
        &=
        \left(
            \frac{|\mu_0|^k \|v_0\|}{\sqrt{|\mu_0|^{2k}\|v_0\| + \| T_\tau^k w_0 \|^2}} - 1
        \right)^2
        \\
        &=
        \left(
            \frac{1}{\sqrt{ 1 + \left( \frac{\| T_\tau w_0 \|}{|\mu_0|^k \|v_0\|} \right)^2}} - 1
        \right)^2
        \\
        &\leq
        \left( 
            \frac{\|T_\tau^k w_0\|}{|\mu_0|^k \|v_0\|}
        \right)^2
        \\
        &=
        \frac{|\mu_1|^{2k}\|w_0\|^2}{|\mu_0|^{2k}\| v_0 \|^2}.
    \end{align*}
    For the inequality above, we used the squared version of the inequality $1 - \frac{1}{\sqrt{1 + x^2}} \leq x$. This can easily be seen to hold true by the fact that the function $f(x)= x+ \frac{1}{\sqrt{1+x^2}} - 1$ satisfies $f(0)=0$ and $f'(x)\geq 0$ for all $x\geq 0$ and thus it holds $f(x)\geq 0$ for all $x \geq 0$. Using the two estimates for the terms above, we get
    \begin{align*}
        \left\|
            \psi_k - \operatorname{sign}(\mu_0)^k\frac{v_0}{\|v_0\|}
        \right\|
        &
        \leq 
        \sqrt{2}\frac{\|w_0\|}{\|v_0\|}\left| \frac{\mu_1}{\mu_0} \right|^k
        \\
        &
        =
        \sqrt{2}\frac{\|w_0\|}{\|v_0\|}\left| \frac{E_0 - \tau}{E_1 - \tau} \right|^k
        \\
        &
        =
        \sqrt{2}\frac{\|w_0\|}{\|v_0\|}\left| 1 - \frac{E_1 - E_0}{E_1 - \tau} \right|^k.
    \end{align*}
\end{proof}

We now turn to the convergence proof for Riemannian gradient descent. As a first step, we prove a descent lemma.

\begin{lemma}[Descent Lemma]\label{eq:descent_lemma}
    For this lemma let $E_0\leq E_1 \leq \dots \leq E_{\textup{max}}$ denote the eigenvalues of $\hat H$ counted with multiplicities and denote by $v_0, v_1,\dots$ an orthonormal basis consisting of eigenvectors. 
    Note that degenerate ground states are still fully covered by this setting.
    Let again $(\psi_k)_{k\in\mathbb N}$ denote the iterates of Riemannian gradient descent with retraction by normalization. Then it holds
    \begin{equation}
        E(\psi_{k+1}) - E(\psi_k)
        =
        -\frac{4\eta}{\|\tilde \psi_{k+1}\|^2}\sum_{i=0}
        (E_i - E(\psi_k))^2\left[ 1 - \eta(E_i - E(\psi_k)) \right]|c_i^{k}|^2,
    \end{equation}
    where 
    \begin{equation*}
        \tilde\psi_{k+1} = \psi_k - \eta\operatorname{grad}_\mathbb SE(\psi_k), 
        \quad \text{and }
        c_i^k = \braket{\psi_k}{v_i}.
    \end{equation*}
    Hence, for $\eta\leq (E_\textup{max} - E_0)^{-1} = \Gamma^{-1}$ it holds that 
    \begin{equation}
        E(\psi_{k+1}) \leq E(\psi_k). 
    \end{equation}
\end{lemma}
\begin{proof}
    First we expand $\psi_k$ and $\tilde \psi_{k+1}$ in the orthonormal basis $(v_i)$
    \begin{equation*}
        \psi_k = \sum_{i=0} c_i^k v_i,
        \quad
        \tilde \psi_{k+1} = \sum_{i=0}\tilde c^{k+1}_i v_i.
    \end{equation*}
    Note that we can write
    \begin{equation*}
        E(\psi_k) = \frac{\sum_{i=0}|\tilde c_i^{k+1}|^2}{\| \tilde \psi_{k+1} \|^2} E(\psi_k),
        \quad
        E(\psi_{k+1}) = \frac{1}{\|\tilde \psi_{k+1}\|^2}\sum_{i=0}E_i|\tilde c_i^{k+1}|^2
    \end{equation*}
    and thus it follows
    \begin{equation*}
        E(\psi_{k+1}) - E(\psi_k)
        =
        \frac{1}{\|\tilde \psi_{k+1}\|^2} \sum_{i=0} (E_i - E(\psi_k))|\tilde c_i^{k+1}|^2
    \end{equation*}
    Note that $\tilde c_i^{k+1}= [1 - 2\eta(E_i - E(\psi_k))]c_i^k$, and inserting in the equality above, using the abbreviation $\Delta_i = E_i - E(\psi_k)$ it follows
    \begin{align*}
        E(\psi_{k+1}) - E(\psi_k)
        &=
        \frac{1}{\|\tilde \psi_{k+1}\|^2}\sum_{i=0}\Delta_i(1 - 2\eta\Delta_i)^2|c_i^k|^2
        \\
        &=
        \frac{1}{\|\tilde \psi_{k+1}\|^2}\sum_{i=0}(\Delta_i - 4\eta\Delta_i^2 + 4\eta^2\Delta_i^3 )|c_i^k|^2
        \\
        &=
        \frac{1}{\|\tilde \psi_{k+1}\|^2}
        \left[
            \sum_{i=0}\Delta_i|c_i^k|^2 - 4\eta\sum_{i=0}\Delta_i^2 |c_i^k|^2 + 4\eta^2\sum_{i=0}\Delta_i^3|c_i^k|^2
        \right].
    \end{align*}
    This is a (weighted) expansion in terms of the first three moments. The first moment vanishes, as can be seen by the following computation
    \begin{equation*}
        \sum_{i=0}\Delta_i|c_i^k|^2 
        =
        \sum_{i=0}(E_i - E(\psi_k))|c_i^k|^2 = E(\psi_k) - E(\psi_k) = 0.
    \end{equation*}
    Rearranging yields the assertion.
\end{proof}

We can now proceed to prove the convergence result for Riemannian gradient descent.

\theoremriemanniangradient*

\begin{proof}
    The proof proceeds in three steps. We begin by showing the invariance of the initial ground state component first, analyze the optimization dynamics along the eigenvectors in a second step, and finally derive the $L^2$ error decay rate.

    \paragraph{First Step: Invariance of the Ground State Direction}
    An un-normalized gradient descent step is of the form
    \begin{equation*}
        \tilde \psi_{k+1} = \psi_k - 2\eta[\hat H \psi_k - E(\psi_k)\psi_k].
    \end{equation*}
    We apply $P_0$ to both sides of this equation and use that $P_0$ commutes with the Hamiltonian
    \begin{align*}
        P_0\tilde\psi_{k+1}
        &=
        P_0\psi_{k} - 2\eta(\hat HP_0\psi_k - E(\psi_k)P_0\psi_k)
        \\
        &=
        [1 - 2\eta(E_0 - E(\psi_k))]P_0\psi_k.
    \end{align*}
    Note that $[1 - 2\eta(E_0 - E(\psi_k))]$ is a positive scalar which we abbreviate with $s_k$. We then compute
    \begin{equation*}
        \frac{P_0\psi_{k+1}}{\| P_0\psi_{k+1} \|}
        =
        \frac{P_0\tilde\psi_{k+1}}{\| P_0\tilde\psi_{k+1} \|}
        =
        \frac{s_kP_0\psi_{k}}{s_k\| P_0\psi_{k} \|}
        =
        \frac{P_0\psi_{k}}{\| P_0\psi_{k} \|}.
    \end{equation*}
    Applying this equality $k$-times yields the invariance
    \begin{equation}\label{eq:auxiliary_invariance_RGD}
        \frac{P_0\psi_{k}}{\| P_0\psi_{k} \|}
        =
        \frac{P_0\psi_{0}}{\| P_0\psi_{0} \|}
        \eqqcolon\psi^*.
    \end{equation}
    
    \paragraph{Second Step: Optimization Dynamics Along Eigenvectors}
    We write $P_0^\perp \psi_k = \sum_{i=d_0}^\text{max}c_i^{k}v_i$, where $d_0=\textup{dim}(V_0)$.
    As a first step we will show 
    \begin{equation}\label{eq:aux_equation_coefficient_ratio}
        \frac{|c_i^{k+1}|}{\|P_0\psi_{k+1}\|}
        =
        \left| 
            \frac{1 - 2\eta(E_i - E(\psi_k))}{1 - 2 \eta(E_0 - E(\psi_k))}
        \right|
        \frac{|c_i^k|}{\|P_0\psi_k\|}.
    \end{equation}
    To show this, we set
    \begin{equation*}
        \tilde \psi_{k+1} = \psi_k - 2\eta(\hat H \psi_k - E(\psi_k)\psi_k)
    \end{equation*}
    and write
    \begin{equation*}
        P_0^\perp\tilde \psi_{k+1} = \sum_{i=1}^{\text{max}}\tilde c_i^{k+1}v_i.
    \end{equation*}
    Note that it holds $\tilde c_i^{k+1}/\|\tilde \psi_{k+1}\| = c_i^{k+1}$. Applying $P_0^\perp$ to $\tilde\psi_{k+1}$ and using that the Hamiltonian commutes with the projection yields
    \begin{equation*}
        P_0^\perp \tilde\psi_{k+1}
        =
        P_0^\perp\psi_k - 2\eta(\hat H P_0^\perp\psi_k - E(\psi_k)P_0^\perp\psi_k).
    \end{equation*}
    Now we take the inner product with the eigenvector $v_i$, which gives
    \begin{align*}
        \tilde c_i^{k+1} &= c_i^k - 2\eta(E_ic_i^k - E(\psi_k)c_i^k)
        \\
        &=
        [1 - 2\eta(E_i - E(\psi_k))]c_i^k.
    \end{align*}
    We also apply $P_0$ to $\tilde\psi_{k+1}$ and obtain
    \begin{equation*}
        P_0\tilde\psi_{k+1} = [1 - 2\eta(E_0 - E(\psi_k))]P_0\psi_k.
    \end{equation*}
    We obtain 
    \begin{align*}
        \frac{|c_i^{k+1}|}{\|P_0\psi_{k+1}\|}
        =
        \frac{|\tilde c_i^{k+1}|\cdot \lVert \tilde{\psi}_{k+1} \rVert}{\|P_0\tilde{\psi}_{k+1}\|\cdot \lVert \tilde{\psi}_{k+1} \rVert}
        =
        \frac{|\tilde c_i^{k+1}|}{\|P_0\tilde\psi_{k+1}\|}
        =
        \left| 
            \frac{1 - 2\eta(E_i - E(\psi_k))}{1 - 2 \eta(E_0 - E(\psi_k))}
        \right|
        \frac{|c_i^k|}{\|P_0\psi_k\|}.
    \end{align*}
    Inserting the step-size $\eta = 1/(2(E_{\text{max}}-E_0))$ we get
    \begin{align*}
        \frac{1 - 2\eta(E_i - E(\psi_k))}{1 - 2\eta(E_0 - E(\psi_k)}
        &=
        1 - \frac{E_i - E_0}{E_{\text{max}} - E_0 + E(\psi_k) - E_0}
        \\
        &\leq
        1 - \frac{E_1 - E_0}{E_{\text{max}} - E_0 + E(\psi_0) - E_0}
        \\
        &\leq
        1 - \frac12\frac{E_1 - E_0}{E_{\text{max}} - E_0} \eqqcolon \rho
    \end{align*}
    where we used $E_i \geq E_1$ and $E(\psi_0) \leq E(\psi_k)$. The latter holds whenever $\eta < 1/(E_\text{max}-E_0)$. Now we use \eqref{eq:aux_equation_coefficient_ratio} $k$-times 
    \begin{equation*}
        \frac{|c_i^k|^2}{\|P_0\psi_k\|^2} \leq \rho^{2k}\frac{|c_i^0|^2}{\|P_0\psi_0\|^2}
    \end{equation*}
    and then sum over $i=1,\dots,\text{max}$ to obtain
    \begin{equation*}
        \frac{\| P_0^\perp\psi_k \|^2}{\|P_0\psi_k\|^2} \leq \rho^{2k}\frac{\| P_0^\perp\psi_0 \|^2}{\|P_0\psi_0\|^2}.
    \end{equation*}

    \paragraph{Convergence in Norm}
    It remains to show that the above implies convergence in norm with the rate $\rho$. We begin by writing 
    \begin{align*}
        \psi_k &= P_0\psi_k + P_0^\perp \psi_k
        \\
        &=
        \|P_0\psi_k\|\psi^* + P_0^\perp\psi_k.
    \end{align*}
    Then, using the above, we compute
    \begin{align*}
        \|\psi_k - \psi^*\|
        &=
        \| (\|P_0\psi_k\| - 1)\psi^* + P_0^\perp\psi_k \|^2
        \\
        &=
        (1 - \|P_0\psi_k\|)^2 + \| P_0^\perp \psi_k \|^2
        \\
        &=
        1 - 2\|P_0\psi_k\| + \|P_0\psi_k\|^2 + 1 - \|P_0\psi_k\|^2
        \\
        &=
        2(1 - \| P_0\psi_k \|)
        \\
        &\leq 
        2(1 - \|P_0\psi_k\|^2)
        \\
        &=
        2\|P_0^\perp \psi_k\|,
    \end{align*}
    where we used the inequality $1-x \leq 1- x^2$ which is valid for $x\in[0,1]$. Finally, we can put it all together by using the final estimate of the second step to obtain
    \begin{equation*}
        \|\psi_k - \psi^*\|^2
        \leq 
        2\|P_0\psi_k\|^2\frac{\|P_0^\perp\psi_0\|^2}{\|P_0\psi_0\|^2}\rho^{2k}
        \leq
        2\frac{\|P_0^\perp\psi_0\|^2}{\|P_0\psi_0\|^2}\rho^{2k}.
    \end{equation*}
\end{proof}

\section{Graphical Visualizations}\label{sec:graphical_visualization}

We provide graphical visualizations of the convergence results for inverse iteration and stochastic reconfiguration. Similar to Figure~\ref{fig:overview}, we use the Hamiltonian
\begin{equation*}
    \hat H = \operatorname{diag}(a, 10, 0):\mathbb R^3 \to \mathbb R^3,
\end{equation*}
where here $a \in(0,10)$ is a free parameter and coincides with the spectral gap $\Delta$ of the Hamiltonian $\hat H$. Figure~\ref{fig:overview} reports the results for $a=1$. The spectral spread $\Gamma=10$ is fixed for all the considered experiments. In the following, we illustrate the effect of suboptimal spectral shifts for PII, the sensitivity of SR to step-size choices, and the effect of closing spectral gaps for SR.

\paragraph{Effect of closing Spectral Gaps}

\begin{figure}
    \centering
    \includegraphics[width=\textwidth]{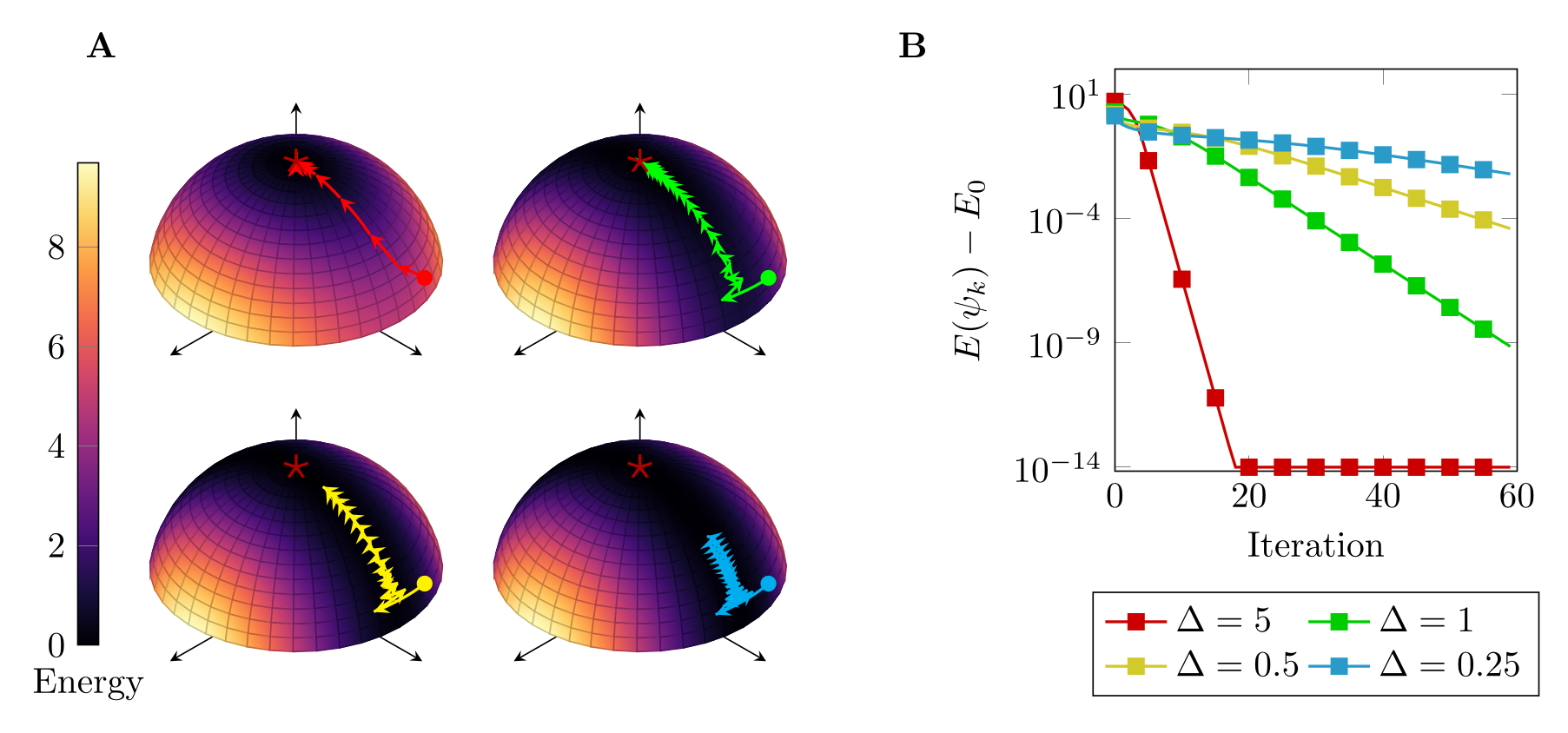}
    \caption{
    Illustration of varying spectral gaps for SR. A smaller spectral gap leads to slower convergence and oscillatory behavior for SR.
    }
    \label{fig:gap_sensitivity}
\end{figure}

In Figure~\ref{fig:gap_sensitivity}, we illustrate the sensitivity of SR to closing spectral gaps. To this end, we vary the parameter $a$. We consider the values $a=5, 1, 0.5$ and $a=0.25$. In this experiment, we choose the learning rate $\eta^*$ to be the (asymptotically) optimal one among the constant learning rate schedules. It is given by
\begin{equation*}
    \eta^* = \frac{1}{E_{\textup{max}} + E_1 - 2 E_0}
\end{equation*}
and solves
\begin{equation*}
    \eta^* = \underset{\eta >0}{\operatorname{argmin}}\left[ \max_{i}|1 - 2\eta(E_i - E_0)| \right]
\end{equation*}
The quantity $|1 - 2\eta(E_i - E_0)|$ is obtained from Eq.~\eqref{eq:aux_equation_coefficient_ratio}, and inserting $E_0$ in place of $E(\psi_k)$. Note that in many scenarios, this step-size is not available as it requires exact knowledge of $E_0$, $E_1$ and $E_{\textup{max}}$. Using $\eta^*$ here however, illustrates the drastically degrading convergence rates of SR a forteriori; even if the optimal step-size was known or found through an expensive hyperparameter tuning, SR can only resolve small-gaped systems through exceedingly expensive computations. This is in direct accordance with Theorem~\ref{thm:convergence_RG}, that states that the error reduction per iteration measured in $L^2$ norm is given by
\begin{equation*}
    \rho_{\text{SR}} = \left| 1 - \frac12\frac{E_1 - E_0}{E_{\textup{max}} - E_0} \right| = \frac{20-a}{20}.
\end{equation*}
Although the above result holds for the slightly different step-size $\eta = 1/(2(E_{\textup{max}}-E_0))$, it describes the qualitative behavior observed in Figure~\ref{fig:gap_sensitivity} well. We used the step-size $\eta^*$ to show that even with an optimal choice, SR can not circumvent a slow-down of convergence in the presence of closing spectral gaps. 

\paragraph{Suboptimal Spectral Shifts in Inverse Iteration}
\begin{figure}
    \centering
    \includegraphics[width=\textwidth]{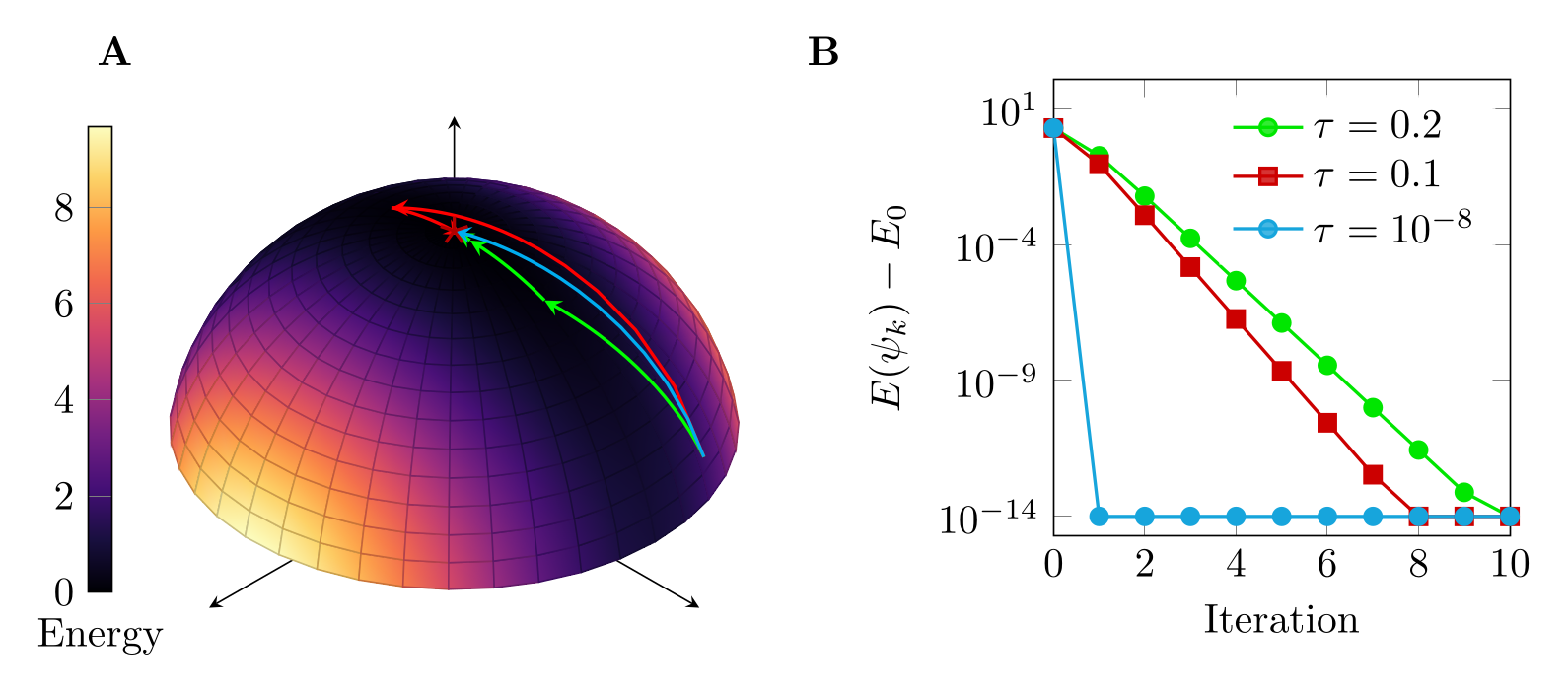}
    \caption{
    Visualization of PII for different shift values $\tau$ for the Hamiltonian
    $\hat H = \operatorname{diag}(1, 10, 0)$ and a linear normalized state vector on the sphere $\mathbb{S}$. Figure A visualizes the optimization dynamics on the sphere, and Figure B shows the corresponding energy error.}
    \label{fig:undershoot}
\end{figure}

We illustrate the effect of gradually varying the shift $\tau$, away from the optimal value $\tau=0$. In Figure \ref{fig:undershoot}, we show the optimization trajectories and the energy errors for $\tau=0.2, 0.1$ and $\tau=10^{-8}$. We clearly see linear convergence in all three cases, with a deteriorating convergence speed as $\tau$ is further away from the optimal value $\tau=0$. This is in agreement with Theorem~\ref{theorem:convergence_shifted_inverse_iteration}, that states that the error reduction per iteration measured in $L^2$ norm is given by ($E_0 = 0$)
\begin{equation*}
    \rho_{\text{PII}} = \left| 1 - \frac{E_1 - E_0}{E_1 - \tau} \right| = \frac{\tau}{1 - \tau}.
\end{equation*}
Section \ref{sec:undershoot} shows that these findings also translate to nonlinearly parametrized wavefunctions, where \emph{undershooting} the optimal $\tau$ also led to an improved robustness of PII in its Monte Carlo variant.

\paragraph{Stability of SR with respect to Step Sizes}
\begin{figure}
    \centering
    \includegraphics[width=\textwidth]{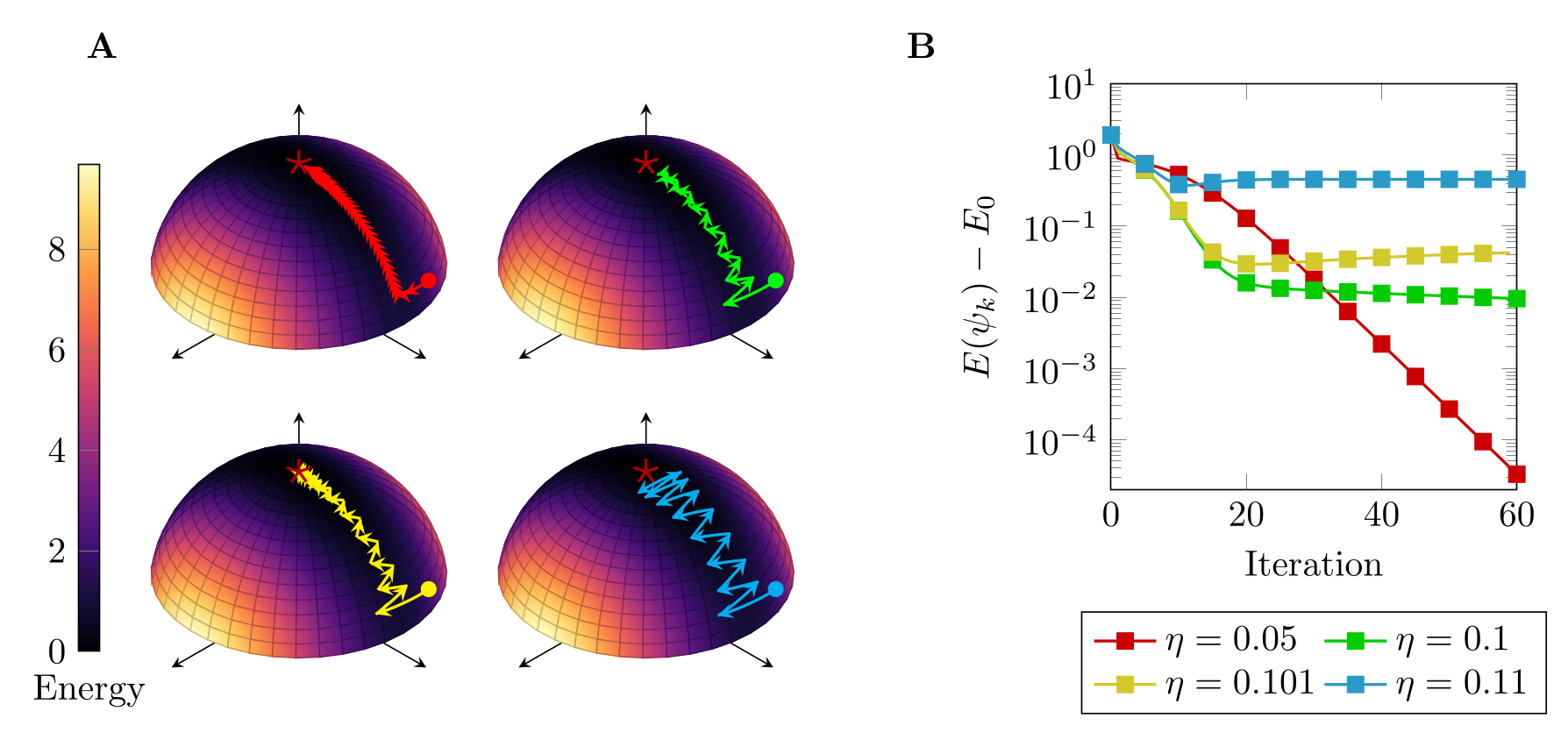}
    \caption{
    Illustration of varying the step sizes for SR. Step-sizes exceeding the critical step-size $\eta=\Gamma^{-1} = 0.1$ prevent the algorithm from convergence.
    }
    \label{fig:stability}
\end{figure}

Finally, we illustrate the role of the critical step-size value $\eta=1/(E_{\textup{max}}-E_0) = 1/\Gamma$. Theorem~\ref{thm:convergence_RG} states that SR converges only if $\eta<1/\Gamma$. To visualize this threshold, we plot the optimization trajectories and corresponding energy errors for the Hamiltonian $\hat H = \operatorname{diag}(1, 10, 0)$ in Figure~\ref{fig:stability}. There, we set $\eta=0.05, 0.1, 0.01$ and $\eta=0.11$. We clearly see that the two choices exceeding the threshold of 0.1 do not converge, and even lead to a non-monotone behavior of the iterates $\{ E(\psi_k) \}$ in the later stages of the optimization. Such a behavior is expected also in the case of more complicated quantum systems and quantum states parametrized as neural networks, and shows that both a large spectral spread $\Gamma$, and a small spectral gap $\Delta$ lead to a slowdown of SR.

\end{appendices}

\bibliography{biblio}

\end{document}